\documentclass{lmcs}
\usepackage[utf8]{inputenc}
\pdfoutput=1

\usepackage{lastpage}
\lmcsdoi{18}{3}{12}
\lmcsheading{}{\pageref{LastPage}}{}{}%
{Aug.~19,~2021}{Jul.~29,~2022}{}


\usepackage{graphicx}
\usepackage[colorlinks=true,linkcolor=green]{hyperref}

\usepackage{amsmath}
\usepackage{amssymb}
\usepackage{mathabx}
\usepackage{mathrsfs}
\usepackage{amsfonts}
\usepackage{tikz}
\usetikzlibrary{automata,positioning,arrows,shapes}
\usepackage{color}
\usepackage[linesnumbered,ruled,vlined]{algorithm2e}
\usepackage{graphicx}
\usepackage{comment}
\usepackage{float}
\usepackage{caption}
\usepackage{booktabs,tabularx}

\usepackage{etoolbox}
\apptocmd{\sloppy}{\hbadness 10000\relax}{}{} 

\tikzset{every state/.style={minimum size=0.5pt}}
\tikzset{every edge/.append style={font=\small}}

\definecolor{myGreen}{RGB}{0,200,0}

\newcommand{\blue}[1]{{\color{black}  #1}}
\newcommand{\green}[1]{{\color{myGreen}  #1}}
\newcommand{\TA}{\mathcal{A}}
\newcommand{\LTS}{\mathcal{T}}
\newcommand\mynodedistance{2.3cm}

\newcommand{\shrink}{\ensuremath{\mathsf{shrink}}}
\newcommand{\enlarge}{\ensuremath{\mathsf{enlarge}}}
\newcommand{\findSSeed}{\ensuremath{\mathsf{findSSeed}}}
\newcommand{\findISeed}{\ensuremath{\mathsf{findISeed}}}

\begin{document}

\title[Timed Automata Robustness Analysis via Model Checking]{Timed Automata Robustness Analysis via Model Checking}

\author[J.~Bend{\'\i}k]{Jaroslav Bend{\'\i}k\lmcsorcid{0000-0001-9784-3028}}[a]	
\author[A.~Sencan]{Ahmet Sencan\lmcsorcid{0000-0003-4275-4655}}[b]	
\author[E.~A.~Gol]{Ebru Aydin Gol\lmcsorcid{0000-0002-5813-9836}}[b]	
\author[I.~\v{C}ern\'a]{Ivana \v{C}ern\'a\lmcsorcid{0000-0002-0711-9552}}[c]	

\address{Max Planck Institute for Software Systems, Kaiserslautern, Germany}	
\email{xbendik@mpi-sws.org}  
\address{Department of Computer Engineering, Middle East Technical University, Ankara, Turkey}	
\email{\{sencan.ahmet, ebrugol\}@metu.edu.tr}  
\address{Faculty of Informatics, Masaryk University, Brno, Czech Republic}	
\email{cerna@fi.muni.cz}  

\begin{abstract}
Timed automata (TA) have been widely adopted as a suitable formalism to model time-critical systems.
Furthermore, contemporary model-checking tools allow the designer to check whether a TA complies with a system specification.
However, the exact timing constants are often uncertain during the design phase. Consequently, the designer is often able to build a TA with a correct structure, however, the timing constants need to be tuned to satisfy the specification. Moreover, even if the TA initially satisfies the specification, it can be the case that just a slight perturbation during the implementation causes a violation of the specification.
Unfortunately, model-checking tools are usually not able to provide any reasonable guidance on how to fix the model in such situations.
In this paper, we propose several concepts and techniques to cope with the above mentioned design phase issues when
dealing
with reachability and safety specifications.

\end{abstract}

\maketitle              
%

\section{Introduction}\label{sec:intro}

Timed automata (TA)~\cite{alur1994theory} extend finite automata with a set of real-time variables, called clocks. The clocks enrich the semantics and the constraints on the clocks restrict the behavior of the automaton, which are particularly important in modeling time-critical systems.   The examples of TA models of critical systems include scheduling of real-time systems~\cite{811256,david2009model,guan2007}, medical devices~\cite{pacemakers:2015,10.1007/s10009-013-0289-7}, rail-road crossing systems~\cite{Wang:2004} and home-care plans~\cite{DBLP:conf/icsoc/GaniBST14}.  

Model-checking methods allow for verifying whether a given TA meets a given system specification. Contemporary model-checking tools, such as UPPAAL~\cite{UPPAAL4:0} or Imitator~\cite{imitatorTool}, have proved to be practically applicable on various industrial case studies~\cite{UPPAAL4:0,Imitatorstudy2019,hytech}.
Unfortunately, during the system design phase, the system information is often incomplete.
A designer is often able to build a TA with correct structure, i.e., exactly capturing locations and transitions of the modeled system, however the exact clock (timing) constraints that enable/trigger the transitions can be uncertain.
Consequently, the produced TA often might not meet the specification (i.e., it does not pass the model-checking) and it needs to be fixed. On the other hand, even though the TA meets the specification, some of its constraints might unnecessarily restrict its behavior or the specification might get violated with a few changes over the constraints (e.g.\ threshold change, constraint removal). In this paper, we present methods to analyze and tune the constraints of the timed automaton to address these issues for reachability and safety specifications. In particular, we identify a minimal set of constraints that needs to be removed to satisfy a reachability specification and we find a maximal set of constraints whose removal do not result in violation of a safety specification. In both cases, we further analyze the thresholds that appear in these constraints.

\paragraph{Tuning TA for a Reachability Specification}
In model-checking, if the considered specification declares an existential property, such as reachability, the property has to hold on a trace of the TA\@. If the property holds, the model checker returns ``yes'' and generates a witness trace satisfying the property. However, if the property does not hold, the model checker usually returns just ``no'' and does not provide any additional information that would help the designer to correct the TA\@.
In this paper, we first study the following problem: given a timed automaton $\TA$ and a reachability property that is not satisfied by $\TA$, relax
 clock constraints of $\TA$ such that the resultant automaton $\TA'$ satisfies the reachability property.
Moreover, the goal is to minimize the number of the relaxed clock constraints and, secondary, also to minimize the overall change of the timing constants used in the clock constraints.
We propose a two step solution for this problem. In the first step, we identify a \emph{minimal sufficient reduction (MSR)} of $\TA$, i.e., an automaton $\TA''$ that satisfies the reachability property and originates from $\TA$ by removing only a minimal necessary set of clock constraints.
In the second step, instead of completely removing the clock constraints, we relax the constraint thresholds. We present two methods for this purpose. First, we employ mixed integer linear programming (MILP) to find a minimal relaxation of the constraints that leads to a satisfaction of the reachability property along a witness path. As the second method, we parametrize the identified constraints and use a parameter synthesis tool to find a minimal parameter valuation such that the target set becomes reachable. The second method is guaranteed to find the globally optimal relaxation as it considers all witness paths, while the first one is more efficient.
We thoroughly compare both the methods on a case study.

The underlying assumption is that during the design the most suitable timing constants reflecting the system properties are defined. Thus, our goal is to generate a TA satisfying the reachability property by changing a minimum number of timing constants.
Some of the constraints of the initial TA can be strict (no relaxation is possible), which can easily be integrated to the proposed solution. Thus, the proposed method can be viewed as a way to handle design uncertainties: develop a TA $\TA$ in a best-effort basis and apply our algorithm to find a $\TA'$ that is \textit{as close as} possible to $\TA$ and satisfies the given reachability property.

\paragraph{Tuning TA for a Safety Specification}
On contrary to existential properties, a universal property, e.g., safety or unavoidability, needs to hold on each trace of the TA\@.
If a safety specification does not hold for the TA, the model-checker returns ``no'' and generates a trace along which the property is violated.
In recent studies, such traces are used to repair the model in an automated way~\cite{10.1007/978-3-030-25540-4_5,ErgurtunaEarlyaccess}.
In the other case, when the safety property holds, the ``yes'' answer obtained from a model-checker simply states that the TA does not have a trace violating the specification. However, the ``yes'' answer does not provide further information on which constraints are effective in avoiding unsafe behaviors and how the verification result changes if some of the constraints are relaxed or removed. The second problem we study aims to provide additional information on a positive verification result for a safety specification described as avoiding a set of ``unsafe'' locations. In particular, we study the following problem: given a timed automaton $\TA$ and a safety property that is satisfied by $\TA$, remove and/or relax the clock constraints of $\TA$ such that the resultant automaton $\TA'$ still satisfies the safety property.
Here, our primary goal is to minimize the number of constraints that need to be left in the TA to prevent reaching the unsafe locations.
Equivalently, we maximize the number of constraints that can be removed from the TA while keeping the unsafe locations unreachable.
Our secondary goal is to maximize the total change in the timing constants used in the remaining clock constraints, where we consider two scenarios: (1) maximize the total change (as in the reachability case) (2) relax each clock constraint with the same amount and maximize this amount. Again, we present a two step solution to the considered problem. In the first step, we identify a \emph{minimal guarantee (MG)} of $\TA$ that is a minimal set of constraints that need to be left in the automaton to ensure that the unsafe locations are still unreachable. In other words the automaton $\TA''$ obtained by removing the constraints that are not in the MG still satisfies the safety specification and removing any additional constraint results in a violation. In the second step, we relax the thresholds of the constraints from the MG (i.e.\ clock constraints of $\TA''$). For both of the aforementioned relaxation scenarios, we parametrize the constraints of $\TA''$ and employ a parameter synthesis tool.

The methods we develop to solve the second problem allows us to relax the TA as much as possible without violating the safety specification. In general, during the design of the automaton, redundant constraints can be added unintentionally to ensure safety. The results of our analysis allows the designer to identify and remove such unnecessary constraints. Furthermore, the constraint constants can be too tight unnecessarily restricting the set of possible behaviors of the automaton. The results obtained in the second step help the designer to relax such constants. On the other hand, if it is not possible to further relax the constraints in the MG, small perturbations results in violation of the specification, in which case, the designer might choose to further restrict some of the constraint constants from the MG\@. Consequently, the developed method is intended to assist the designer to improve the model that is generated in a best-effort manner.

\blue{The proposed approach for tuning a TA for reachability specifications first appeared in~\cite{DBLP:conf/tacas/BendikSGC21}. This paper extends~\cite{DBLP:conf/tacas/BendikSGC21} by introducing the minimal guarantee concepts and the corresponding methods to (\textit{i}) generate MG and (\textit{ii}) the corresponding relaxations for tuning TA for safety specifications.}

\paragraph{Outline.} The rest of the paper is organized as follows. Section~\ref{sec:prelims} introduces basic concepts used throughout the paper and formally defines the problems we deal with. Subsequently, in Sections~\ref{sec:msr} and~\ref{sec:mg}, we describe our approaches for identifying Minimal Sufficient Reductions (MSRs) and Minimal Guarantees (MGs), respectively. In Section~\ref{sec:msr-relax}, we describe how to just relax timing constraints in an MSR instead of completely removing the constraints from the TA\@. Similarly, in Section~\ref{sec:mg-relax} we show how to further relax an MG via parameter synthesis.
In Section~\ref{sec:related}, we provide an overview of related work.
Finally, we experimentally evaluate the proposed techniques in Section~\ref{sec:experiments}, and conclude in Section~\ref{sec:conclusion}.

\section{Preliminaries}\label{sec:prelims}

\subsection{Timed Automata}\label{sec:prelims:ta}

A \emph{timed automaton} (TA)~\cite{alur1999timed,alur1994theory,larsen1993time} is a finite-state machine extended with a finite set $C$ of real-valued clocks. A \emph{clock} $x\in C$ measures the time spent after its last reset. In a TA, clock constraints are defined for locations (states) and transitions. A \emph{simple clock constraint} is defined as $ x - y \sim c$ where $x, y \in C \cup \{0\}$, $\sim \in \{<, \leq \}$  and $c \in \mathbb{Z} \cup \{\infty\}$.%
\footnote{Simple constraints are only defined as upper bounds to simplify the presentation. This definition is not restrictive since $x - y \geq c$ and $x \geq c$ are equivalent to $y - x \leq - c$ and $0 - x \leq -c$, respectively. A similar argument holds for strict inequality $(>)$.}
Simple clock constraints and constraints obtained by combining these with the conjunction operator ($\wedge$) are called \emph{clock constraints}.
The sets of simple and all clock constrains  are denoted by $\Phi_S(C)$ and $\Phi(C)$, respectively. For a clock constraint $\phi \in \Phi(C)$,  $\mathcal{S}(\phi)$ denotes the simple constraints from $\phi$, e.g., $\mathcal{S}(x - y < 10 \wedge y \leq 20) = \{x - y < 10, y \leq 20\}$.
A clock constraint is called \emph{parametric} if the numerical constant (i.e. $c$) is represented by a parameter.
A \emph{clock valuation} $v: C \to \mathbb{R}_+$ assigns non-negative real values to each clock.
The notation $v \models \phi$ denotes that the clock constraint $\phi$ evaluates to true when each clock $x$ is replaced with  $v(x)$.
For a clock valuation $v$ and $ d \in \mathbb{R}_+$, $v+d$ is the clock valuation obtained by \emph{delaying} each clock by $d$, i.e., $(v+d)(x) = v(x) +d $ for each $x\in C$. For $\lambda \subseteq C$, $v[\lambda := 0]$ is the clock valuation obtained after \emph{resetting} each clock from $\lambda$, i.e.,   $v[\lambda := 0](x) = 0$ for each $x\in \lambda$ and $v[\lambda := 0](x) = v(x)$ for each $x\in C\setminus \lambda$.

\begin{defi}[Timed Automata]
A \textit{timed automaton}  $\TA = (L, l_0, C, \Delta, Inv )$ is a tuple, where
$L$ is a finite set of locations,
$l_{0} \in L$ is the initial location,
$C$ is a finite set of clocks,
$\Delta \subseteq L \times 2^C \times  \Phi(C) \times L$ is a finite transition relation, and
$Inv: L \to \Phi(C)$ is an invariant function.
\end{defi}

For a  transition $e = (l_s, \lambda, \phi, l_t) \in \Delta$, $l_s$ is the source location, $l_t$ is the target location, $\lambda$ is the set of clocks reset on $e$ and $\phi$ is the guard (i.e., a clock constraint) tested for enabling $e$.
A Parametric TA (PTA) extends TA by allowing the use of parametric constraints. Given a PTA $\TA$ with parameter set $P$ and a parameter valuation $\textbf{p} : P \to \mathbb{N}$, a (non-parametric) TA $\TA(\textbf{p})$ is obtained by replacing each parameter $p \in P$ in $\TA$ with the corresponding valuation $\textbf{p}(p)$.
The semantics of a TA is given by a labelled transition system (LTS). An LTS is a tuple $\LTS = (S, s_0, \Sigma, \to)$, where $S$ is a set of states, $s_0 \in S$ is an initial state, $\Sigma$ is a set of symbols, and $\to\ \subseteq S \times \Sigma \times S$ is a transition relation. A transition $(s,a,s') \in\ \to$ is also denoted as $s \stackrel{a}{\to} s'$.

\begin{defi}[LTS semantics for TA] Given a TA $\TA = (L, l_0, C, \Delta, Inv )$, the labelled transition system $T(\TA) = (S, s_0, \Sigma, \to)$ is defined as follows:
\begin{itemize}
\item $S = \{(l, v) \mid l \in L, v\in \mathbb{R}_{+}^{|C|}, v \models Inv(l)\}$,
\item $s_0 = (l_0, \textbf{0})$, where $ \textbf{0}(x) = 0$ for each $x\in C$,
\item $\Sigma = \{act\} \cup  \mathbb{R}_{+}$, and
\item the transition relation $\to$ is defined by the following rules:
\begin{itemize}
\item delay transition: $(l, v) \stackrel{d}{\to} (l, v+d)$ if $v+d \models Inv(l)$
\item discrete transition: $(l, v) \stackrel{act}{\to} (l', v')$ if there exists $(l,\lambda, \phi, l') \in \Delta$ such that $v \models \phi$, $v' = v[\lambda:=0]$, and $v' \models Inv(l')$.
\end{itemize}
\end{itemize}
\end{defi}

\noindent
The notation $s \mathrel{{\to}_d} s'$ is used to denote a delay transition of duration $d$ followed by a discrete transition from $s$ to $s'$, i.e., $s  \stackrel{d}{\to} s  \stackrel{act}{\to} s'$.  A run $\rho$ of $\TA$ is  either a finite or an infinite alternating sequence of delay and discrete transitions, i.e.,
$\rho = s_0  {\to}_{d_0} s_1 {\to}_{d_1} s_2  {\to}_{d_2} \cdots$.
The set of all runs of $\TA$ is denoted by $[[ \TA ]]$.

A path $\pi$ of $\TA$ is an interleaving sequence of locations and transitions, $\pi= l_0, e_1, l_1, e_2, \ldots$, where $e_{i+1}  = (l_{i}, \lambda_{i+1}, \phi_{i+1}, l_{i+1})\in \Delta$ for each $i \geq 0$.
A path $\pi= l_0, e_1, l_1, e_2, \ldots$ is
\textit{realizable} if there exists a delay sequence $d_0, d_1, \ldots$ such that $(l_0, \textbf{0}) {\to}_{d_0} (l_1,  v_1) {\to}_{d_1} (l_1,  v_2 ) {\to}_{d_2} \cdots$ is a run of $\TA$ and for every $i \geq 1$, the $i$th discrete transition is taken according to $e_i$, i.e., $e_i = (l_{i-1}, \lambda_i, \phi_i, l_i)$, $v_{i-1} + d_{i-1} \models \phi_i$,  $v_i = (v_{i-1} + d_{i-1})[\lambda_i:=0]$ and $v_i \models Inv'(l_i)$.

For a TA $\TA$ and a subset of its locations $L_T \subseteq L$, $L_T$ is said to be \textit{reachable} on $\TA$ if there exists $\rho = (l_0, \textbf{0})  {\to}_{d_0} (l_1, v_1) {\to}_{d_1}  \ldots {\to}_{d_{n-1}} (l_n, v_n)\in [[\TA]]$ such that $l_n \in L_T$; otherwise, $L_T$ is \textit{unreachable}.
In this study, $L_T$ is used to denote the set of \emph{target locations} for reachability specifications and the set of \emph{unsafe locations} for safety specifications.
In the latter case, $\TA$ is called \emph{safe} if $L_T$ is unreachable; otherwise $\TA$ is \emph{unsafe}.
The reachability problem, $isReachable(\TA , L_T)$, is decidable and implemented in various verification tools including UPPAAL~\cite{UPPAAL4:0}. The verifier either returns ``No'' indicating that such a run does not exist, or it generates a run
(counter-example) leading from the initial state of $\TA$ to a location $v_n \in L_T$.

\begin{exa}\label{ex:running}
In Figure~\ref{fig:ex}, we illustrate a TA with 8 locations: $\{l_0,\ldots, l_7\}$, 9 transitions: $\{e_1, \ldots, e_9\}$, an initial location $l_0$, and a set of unreachable locations $L_T =\{l_4\}$.
\end{exa}

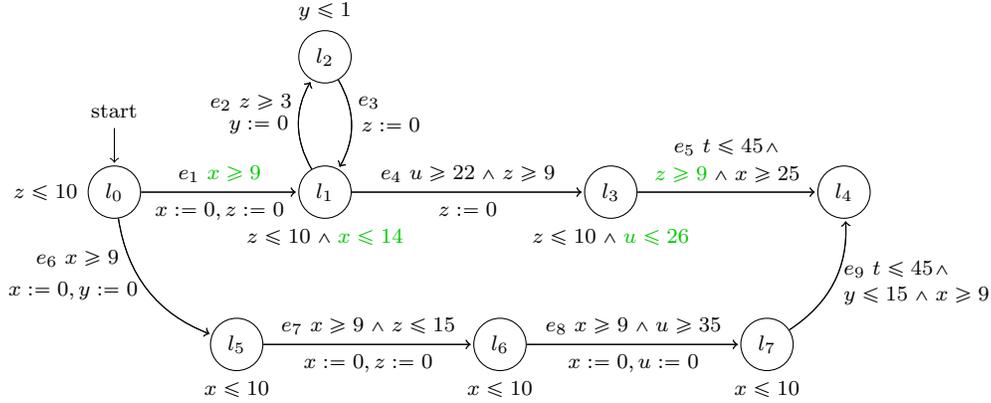
\begin{figure}[t]
\begin{center}
\begin{tikzpicture}[shorten >=1 pt, node distance = \mynodedistance, on
grid, auto]
\tikzstyle{every node}=[font=\scriptsize]
\node [state, initial,initial where=above, label = {left: $z \leq 10$}, initial where=left] (l_0) {$l_0$};
\node [state] (l_1) [right = of l_0, label = {below: $z \leq 10\wedge \green{x \leq 14}$}, xshift=0.5cm] {$l_1$};
\node [state] (l_2) [above = of l_1, label = {above: $y \leq  1$}, yshift=-0.5cm] {$l_2$};
\node [state] (l_3) [right = of l_1, label = {below: $ z \leq 10\wedge \green{u \leq 26}$}, xshift=1.5cm] {$l_3$};
\node [state] (l_4) [right = of l_3, xshift=0.8cm] {$l_4$};
\node [state] (l_5) [below right = of l_0, label = {below: $x \leq 10$}, yshift=-0.4cm] {$l_5$};
\node [state] (l_6) [right = of l_5, label = {below: $x \leq 10$},xshift=1.2cm] {$l_6$};
\node [state] (l_7) [right = of l_6, label = {below: $x \leq 10$},xshift=1.25cm] {$l_7$};
\path [->]
(l_0) edge node [below] {$x := 0, z:=0$} (l_1)
(l_0) edge node [above] {$e_1\    \green{x \geq 9} $} (l_1)
(l_1) edge [bend left] node [left, pos=0.75] {$e_2$  $z \geq3 $} (l_2)
(l_1) edge [bend left] node [left, pos=0.5] {$y := 0  $} (l_2)
(l_2) edge [bend left] node [right, pos=0.25] {$e_3 $} (l_1)
(l_2) edge [bend left] node [right, pos=0.5] {$z := 0  $} (l_1)
(l_1) edge node [above] {$e_4\ u \geq 22 \wedge z \geq 9$} (l_3)
(l_1) edge node [below] {$z:=0$} (l_3)
(l_3) edge node [above] {$\green{z\geq9}\wedge x\geq25$} (l_4)
(l_3) edge node [above,yshift=10] {$e_5\ t\leq45\wedge$} (l_4)
(l_0) edge [bend right] node [left, pos=0.25] {$e_6\ x \geq 9$} (l_5)
(l_0) edge [bend right] node [left, pos=0.5] {$x := 0, y:=0$} (l_5)
(l_5) edge node [above] {$e_7\ x \geq 9 \wedge z \leq 15$} (l_6)
(l_5) edge node [below] {$x:=0, z:=0$} (l_6)
(l_6) edge node [above] {$e_8\ x \geq 9 \wedge u \geq 35 $} (l_7)
(l_6) edge node [below] {$x:=0,u:=0$} (l_7)
(l_7) edge [bend right]node [right,yshift=-5] {$ y \leq 15\wedge x\geq9$} (l_4)
(l_7) edge [bend right]node [right,yshift=5] {$e_9 \ t\leq45\wedge$} (l_4)
;
\end{tikzpicture}
\end{center}
\caption{An illustration of a TA used in Examples~\ref{ex:running},~\ref{ex:msr} and~\ref{ex:msr2}.}\label{fig:ex}
\end{figure}

\subsection{Timed Automata Relaxation}\label{sec:prelims-relaxation}

For a timed automaton $\TA = (L, l_0, C, \Delta, Inv )$, the set of pairs of transition and associated simple constraints is defined in~\eqref{eq:guard_cons_set} and the set of pairs of location and associated simple constraints is defined in~\eqref{eq:inv_cons_set}.
\begin{align}
 \Psi(\Delta) &= \{ (e, \varphi)  \mid  e = (l_s, \lambda, \phi, l_t) \in \Delta, \varphi \in \mathcal{S}(\phi) \}\label{eq:guard_cons_set} \\
 \Psi(Inv) &= \{ (l, \varphi)  \mid  l \in L, \varphi \in \mathcal{S}(Inv(l)) \}\label{eq:inv_cons_set}
\end{align}

\begin{defi}[constraint-relaxation]\label{defn:constraint_relaxation}Let $\phi \in \Phi(C)$ be a constraint over $C$, $\Theta \subseteq \mathcal{S}(\phi)$ be a subset of its simple constraints and $\mathbf{r} : \Theta \to \mathbb{N}\cup \{\infty\}$ be a positive valued relaxation valuation. The relaxed constraint is defined as:
\begin{equation}
R (\phi, \Theta, \mathbf{r}) = \left ( \bigwedge_{\varphi \in \mathcal{S}(\phi) \setminus \Theta} \varphi  \right ) \wedge \left ( \bigwedge_{\varphi  = x - y \sim c \in \Theta} x - y \sim c + \mathbf{r}(\varphi) \right )
\end{equation}
\end{defi}
Intuitively, $R (\phi, \Theta, \mathbf{r})$ relaxes only the thresholds of simple constraints from $\Theta$ with respect to $\mathbf{r}$, e.g., $R( x - y \leq 10 \wedge y < 20, \{y < 20\},  \mathbf{r}) = x - y \leq 10 \wedge y < 23$, where $\mathbf{r}(y<20)=3$.  Setting a threshold to $\infty$ implies removing the corresponding  simple constraint, e.g., $R( x - y \leq 10 \wedge y < 20, \{y < 20\},\mathbf{r}) = x - y \leq 10$, where $\mathbf{r}(y<20)=\infty$. Note that $R (\phi, \Theta, \mathbf{r}) = \phi$ when $\Theta$ is empty.

\begin{defi}[($D,I,\mathbf{r}$)-relaxation]\label{defn:relaxation}
Let $\TA = (L, l_0, C, \Delta, Inv )$ be a TA, $D \subseteq \Psi(\Delta)$ and $I \subseteq \Psi(Inv)$ be transition and location constraint sets, and $\mathbf{r}: D  \cup I \to \mathbb{N}\cup \{\infty\}$ be a positive valued relaxation valuation.
The ($D,I,\mathbf{r}$)-relaxation of $\TA$, denoted $\TA_{<D,I,\mathbf{r}>}$, is a TA $\TA' = (L', l_0', C', \Delta', Inv' )$ such that:

\begin{itemize}
\item $L = L'$, $l_0 = l_0'$, $C = C'$, and
\item $\Delta' $ originates from $\Delta$ by relaxing $D$ via $\mathbf{r}$.  For $e = (l_s, \lambda, \phi, l_t) \in \Delta$, let $D |_{e} = \{\varphi \mid (e, \varphi) \in D\}$, and let $\mathbf{r}|_{e}(\varphi) = \mathbf{r}(e, \varphi)$, then
$\Delta' = \{  (l_s, \lambda, R(\phi, D |_{e}, \mathbf{r}|_{e}), l_t) \mid e = (l_s, \lambda, \phi, l_t) \in \Delta\}$
\item $Inv'$ originates from $Inv$ by relaxing $I$ via $\mathbf{r}$. For $l \in L$, let $I|_l = \{\varphi \mid (l, \varphi) \in I\}$, and $\mathbf{r}|_{l}(\varphi) = \mathbf{r}(l, \varphi)$, then $Inv'(l) = R(Inv(l), I |_{l}, \mathbf{r}|_{l})$.
\end{itemize}
\end{defi}

\noindent
Intuitively, the TA $\TA_{<D,I,\mathbf{r}>}$ emerges from $\TA$ by relaxing the guards of the transitions from the set $D$ and relaxing invariants of the locations from $I$ with respect to $\mathbf{r}$.

\begin{prop}\label{proposition:run-inclusion}
Let $\TA = (L, l_0, C, \Delta, Inv )$ be a timed automaton, $D \subseteq \Psi(\Delta)$ and $I \subseteq \Psi(Inv)$ be sets of simple guard and invariant constraints, and $\mathbf{r}: D  \cup I \to \mathbb{N}\cup \{\infty\}$ be a relaxation valuation. Then $[[\TA]] \subseteq [[\TA_{<D,I,\mathbf{r}>}]]$.
\end{prop}

\begin{proof}
Observe that for a clock constraint $\phi \in \Phi(C)$, a subset of its simple constraints $\Theta \subseteq \mathcal{S}(\phi)$, a relaxation valuation $\mathbf{r}'$ for $\Theta$, and
the relaxed constraint $R (\phi, \Theta, \mathbf{r}')$ as in Definition~\ref{defn:constraint_relaxation}, it holds  that for any clock valuation $v: v \models \phi \implies v \models  R (\phi, \Theta, \mathbf{r}')$.
Now, consider a run $\rho = (l_0, \textbf{0})  {\to}_{d_0} (l_1, v_1) {\to}_{d_1} (l_2, v_2)  {\to}_{d_2} \cdots \in[[\TA]]$.
Let $\pi= l_0, e_1, l_1, e_2, \ldots$ with $e_{i}  = (l_{i-1}, \lambda_{i}, \phi_{i}, l_{i})\in \Delta$ for each $i \geq 1$ be the path realized as $\rho$ via delay sequence $d_0, d_1, \ldots$.
By Definition~\ref{defn:relaxation} for each $ (l, \lambda, \phi, l') \in \Delta$, there is $(l, \lambda, R(\phi, D |_{e}, \mathbf{r}|_{e}), l')\in \Delta'$.
We define a path induced by $\pi$ on $\TA_{<D,I,\mathbf{r}>}$ as:
\begin{equation}\label{eq:path_map}
M(\pi) = l_0, (l_{0}, \lambda_{1}, R(\phi_{1}, D |_{e_1}, \mathbf{r}|_{e_1}), l_{1}), l_1,  (l_{1}, \lambda_{2}, R(\phi_{2}, D |_{e_2}, \mathbf{r}|_{e_2}), l_2) , \ldots
\end{equation}
For each $i=0,\ldots,n-1$ it holds that $v_i \models R(Inv(l_i), D |_{l_i} ,  \mathbf{r}|_{l_i})$, $v_i + d_i\models R(Inv(l_i), D |_{l_i} ,  \mathbf{r}|_{l_i})$ and $v_i + d_i \models R(\phi_{i+1}, D |_{e_{i+1}}, \mathbf{r}|_{e_{i+1}})$. Thus $M(\pi)$ is realizable on
$\TA_{<D,I,\mathbf{r}>}$ via the same delay sequence and $\rho \in [[\TA_{<D,I,\mathbf{r}>}]]$. As $\rho \in [[\TA]]$ is arbitrary, we conclude that $[[\TA]] \subseteq [[\TA_{<D,I,\mathbf{r}>}]]$.
\end{proof}

\subsection{Reductions and Guarantees}\label{sec:prelims-reductions}

\begin{defi}\label{def:reduction}
A \emph{reduction} is a relaxation $\TA_{<D,I,\mathbf{r}>}$ of $\mathcal{A}$ such that $\mathbf{r}(a) = \infty$ for each $a\in D\cup I$.
Moreover, since $\mathbf{r}$ is fixed, we simply denote the reduction by $\TA_{<D,I>}$.
\end{defi}

Intuitively, a reduction $\TA_{<D,I>}$ effectively removes all the simple constraints $D \cup I$ from $\mathcal{A}$. Also, note that $\TA = \TA_{<\emptyset,\emptyset>}$. Hereafter, we use two notations for naming a reduction; either we simply use capital letters, e.g., $M,N,K$ to name a reduction, or we use the notation $\TA_{<D,I>}$ to also specify the sets $D,I$ of simple clock constraints.
Given a reduction $N = \TA_{<D,I>}$, $|N|$ denotes the cardinality $|D \cup I|$.
Furthermore,  $\mathcal{R}_{\TA}$ denotes the set of all reductions of $\mathcal{A}$.
We define a partial order relation $\sqsubseteq$ on $\mathcal{R}_{\TA}$ as $\TA_{<D,I>} \sqsubseteq \TA_{<D',I'>}$ iff $D \cup I \subseteq D' \cup I'$. Similarly, we write $\TA_{<D,I>} \sqsubsetneq \TA_{<D',I'>}$ iff $D \cup I \subsetneq D' \cup I'$. We say that a reduction $\TA_{<D,I>}$ is a \emph{sufficient reduction} (w.r.t. $\TA$ and $L_T$) iff $L_T$ is reachable on $\TA_{<D,I>}$; otherwise, $\TA_{<D,I>}$ is an \emph{insufficient reduction}. Crucially, observe that the property of being a sufficient reduction is monotone w.r.t.\ the partial order:

\begin{prop}\label{proposition:monotonicity}
Let $\TA_{<D,I>}$ and $\TA_{<D',I'>}$ be reductions such that $\TA_{<D,I>} \sqsubseteq \TA_{<D',I'>}$. If $\TA_{<D,I>}$ is sufficient then $\TA_{<D',I'>}$ is also sufficient.
\end{prop}

\begin{proof}
Note that
$\TA_{<D',I'>}$ is a ($D' \setminus D$,$I' \setminus I$)-reduction of $\TA_{<D,I>}$. By Proposition~\ref{proposition:run-inclusion},
$[[\TA_{<D,I>}]] \subseteq [[\TA_{<D',I'>}]]$, i.e., the run of $\TA_{<D,I>}$ that witnesses the reachability of $L_T$ is also a run of  $\TA_{<D',I'>}$.
\end{proof}

\begin{defi}[MSR]\label{defn:msr}
A sufficient reduction $\TA_{<D,I>}$ is a \emph{minimal sufficient reduction (MSR)} iff there is no
$c \in D \cup I$ such that the reduction $\TA_{<D \setminus \{c\},I  \setminus \{c\}>}$ is sufficient.
Equivalently, due to Proposition~\ref{proposition:monotonicity}, $\TA_{<D,I>}$ is an MSR iff there is no sufficient reduction $\TA_{<D',I'>}$ such that $\TA_{<D',I'>} \sqsubsetneq \TA_{<D,I>}$.
\end{defi}

\begin{defi}[MIR]\label{defn:mir}
An insufficient reduction $\TA_{<D,I>}$ is a \emph{maximal insufficient reduction (MIR)} iff there is no
$c \in (\Psi(\Delta) \cup \Psi(Inv)) \setminus (D \cup I)$ such that the reduction $\TA_{<D', I'>}$ with $D' \cup I' = D \cup I \cup \{c\}$ is insufficient.
Equivalently, due to Proposition~\ref{proposition:monotonicity}, $\TA_{<D,I>}$ is an MIR iff there is no insufficient reduction $\TA_{<D'',I''>}$ such that $\TA_{<D,I>} \sqsubsetneq \TA_{<D'',I''>}$.
\end{defi}

Intuitively, an MSR represents a minimal set of constraints that need to be removed from $\mathcal{A}$ to make the target location(s) $L_T$ reachable, whereas an MIR represents a maximal set of constraints whose removal does not make the target location(s) reachable.

Recall that a reduction $\TA_{<D,I>}$ is determined by $D \subseteq \Psi(\Delta)$ and $I \subseteq \Psi(Inv)$. Consequently, $|\mathcal{R}_{\TA}| = 2^{|\Psi(\Delta) \cup \Psi(Inv)|}$ (i.e., there are exponentially many reductions w.r.t. $|\Psi(\Delta) \cup \Psi(Inv)|$). Moreover, there can be up to $\binom{k}{k/2}$ MSRs (MIRs) where $k = |\Psi(\Delta) \cup \Psi(Inv)|$.\footnote{There are $\binom{k}{k/2}$ pair-wise incomparable elements of $\mathcal{R}_{\TA}$ w.r.t. $\sqsubsetneq$ (see Sperner’s theorem~\cite{sperner1928satz}) and all of them can be MSRs (or MIRs).}
Also note, that the \emph{minimality} (\emph{maximality}) of a reduction does not mean a \emph{minimum} (\emph{maximum}) number of simple clock constraints that are removed by the reduction; there can exist two MSRs (MIRs), $M$ and $N$, such that $|M| < |N|$.
We call an MSR $M$ a \emph{minimum MSR} if there is no MSR $M'$ with $|M'| < |M|$.
Similarly, an MIR $M$ is a \emph{maximum MIR} if there is no MIR $M'$ with $|M'| > |M|$.
Note that there can be also up to $\binom{k}{k/2}$ minimum MSRs and up to $\binom{k}{k/2}$ maximum MIRs.

In some applications, instead of thinking about an MIR of $\mathcal{A}$, i.e.\ a maximal set of simple clock constraints whose removal does not make the target location(s) reachable, it might be more natural to think about the complement of an MIR, i.e, a \emph{minimal set of simple clock constraints} that need to be left in $\mathcal{A}$ to ensure that the target location is still unreachable. We define this complementary notion as a \emph{minimal guarantee}:

\begin{defi}[MG]\label{defn:mg}
Given a reduction $\TA_{<D,I>}$, the set $D \cup I$ of simple clock constraints constitutes a \emph{guarantee} (for $\mathcal{A}$) iff
the reduction $\TA_{<\Psi(\Delta) \setminus D, \Psi(Inv) \setminus I>}$ is insufficient.
Furthermore, a guarantee $D \cup I$ is \emph{a minimal guarantee} (MG) iff
for every $c \in D \cup I$ the reduction $\TA_{<\Psi(\Delta) \setminus (D \setminus \{c\}) , \Psi(Inv) \setminus (I \setminus \{c\})>}$ is sufficient.
Equivalently, due to Proposition~\ref{proposition:monotonicity}, $D \cup I$ is an MG iff there is no guarantee $D' \cup I'$ such that $D' \cup I' \subsetneq D \cup I$.
\end{defi}

\begin{obs}\label{obs:mir-mg-complement}
A reduction $\TA_{<D,I>}$ is insufficient iff the set $(\Psi(\Delta) \cup \Psi(Inv)) \setminus (D \cup I)$ is a guarantee. Furthermore,
$\TA_{<D,I>}$ is an MIR iff $(\Psi(\Delta) \cup \Psi(Inv)) \setminus (D \cup I)$ is an MG\@.
\end{obs}

Note that due to technical reasons, we define the concept of an MG as \emph{a set} $D \cup I$ \emph{of simple clock constraints}, whereas the concept of an MIR is defined as a \emph{reduction}
$\TA_{<D',I'>}$ (i.e., a TA) that is determined by a set $D' \cup I'$ of simple clock constraints.

\begin{exa}\label{ex:msr}
Assume the TA $\TA$ and $L_T = \{l_4\}$ from Example~\ref{ex:running} (Fig.~\ref{fig:ex}).
There are 24 MSRs and 4 of them are minimum. For example, $\TA_{<D,I>}$ with $D=\{ (e_5, x \geq 25) \}$ and $I = \{ (l_3, u \leq 26) \}$ is a minimum MSR, and $\TA_{<D',I'>}$ with $D' = \{(e_9, y \leq 15), (e_7, z \leq 15) \}$ and $I' = \{ (l_6, x \leq 10) \}$
is a non-minimum MSR\@.
There are 40 MGs (and hence MIRs) and 21 of them are minimum. For instance, $D \cup I$ with
$D = \{ (e_5, z \geq 9), (e_8, u \geq 35), (e_1, x \geq 9), (e_4, z \geq 9) \}$ and
$I = \{ (l_5, x \leq 10), (l_6, x \leq 10), (l_0, z \leq 10), (l_3, u \leq 26) \}$
is a non-minimum MG, and $D' \cup I'$ with $D' = \{ (e_7, x \geq 9), (e_9, y \leq 15), (e_8, x \geq 9), (e_5, x \geq 25) \}$ and $I' = \{ (l_1, x \leq 14), (l_3, z \leq 10) \}$ is a minimum MG\@.
\end{exa}

\blue{Finally, note that in some situations, we might not want to include the whole set $\Psi(\Delta) \cup \Psi(Inv)$ of all simple clock constraints in the analysis but rather just its subset (e.g., because some simple clock constraints simply could not be modified). Our definitions of ($D,I,\mathbf{r}$)-relaxations, reductions, and guarantees, can be naturally extended also to work with just a subset of $\Psi(\Delta) \cup \Psi(Inv)$. We illustrate this on a simple example.}

\blue{
\begin{exa}\label{ex:msr2}
Assume that only 4 of the simple clock constraints from the TA in Fig.~\ref{fig:ex} can be removed/relaxed and the other simple clock constraints represent physical limitations that can not be changed. The tunable simple clock constraints are $c_1 = x \geq 9$, $c_2 = z \geq 9$, $c_3 = x \leq 14$ and $c_4 = u \leq 26$ that appear on edge $e_1$, edge $e_5$, location $l_1$, and location $l_3$, respectively. Note that these constraints are highlighted using green color in Fig.~\ref{fig:ex}. If we restrict our analysis only to those four constraints, then there are 2 MSR:\@
$\{c_3, c_4\}$ and $\{c_1, c_2, c_3\}$, and three MGs:
$\{c_3\}$, $\{c_1, c_4\}$ and $\{c_2, c_4\}$. We provide a power-set illustration of this example in Fig.~\ref{fig:ex-powerset}
\end{exa}
}

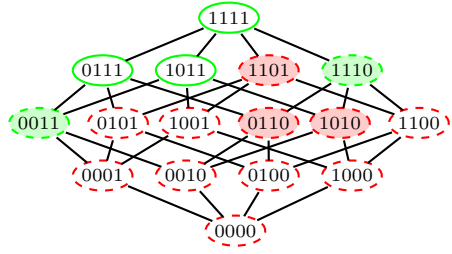
\begin{figure}[t!]
\centering
\begin{tikzpicture}[-,>=stealth',shorten >=1pt,auto,node distance=1.5cm,minimum size=0.4cm,
thick,every node/.style={draw, ellipse,inner sep=0, outer sep=0},
valid/.style={draw,draw=red,dashed},
invalid/.style={draw,draw=green},
mus/.style={draw,draw=green,fill=green!20,dashed},
mss/.style={draw,draw=red,fill=red!20,dashed},
v/.style={fill=green},
i/.style={fill=red}]
\node[valid] (0) []{\tiny 0000};
\node[valid] (01) [above left = 0.44999999999999996 and -2.11 of 0]{\tiny 1000};
\node[valid] (02) [above left = 0.44999999999999996 and -1.005 of 0]{\tiny 0100};
\node[valid] (03) [above left = 0.44999999999999996 and 0.1 of 0]{\tiny 0010};
\node[valid] (04) [above left = 0.44999999999999996 and 1.205 of 0]{\tiny 0001};

\node[valid] (012) [above left = 1.15 and -2.9999999999999996 of 0]{\tiny 1100};
\node[mss] (013) [above left = 1.15 and -1.9499999999999998 of 0]{\tiny 1010};
\node[mss] (023) [above left = 1.15 and -1 of 0]{\tiny 0110};
\node[valid] (014) [above left = 1.15 and 0.05 of 0]{\tiny 1001};
\node[valid] (024) [above left = 1.15 and 0.9999999999999999 of 0]{\tiny 0101};
\node[mus] (034) [above left = 1.15 and 2.0500000000000001 of 0]{\tiny 0011};

\node[mus] (0123) [above left = 1.8499999999999996 and -2.11 of 0]{\tiny 1110};
\node[mss] (0124) [above left = 1.8499999999999996 and -1.005 of 0]{\tiny 1101};
\node[invalid] (0134) [above left = 1.8499999999999996 and 0.1 of 0]{\tiny 1011};
\node[invalid] (0234) [above left = 1.8499999999999996 and 1.205 of 0]{\tiny 0111};
\node[invalid] (01234) [above left = 2.55 and -0.475 of 0]{\tiny 1111};
\path[every node/.style={font=\sffamily\small}]
(0) edge[] node [left] {} (01)
    edge[] node [left] {} (02)
    edge[] node [left] {} (03)
    edge[] node [left] {} (04)
(01) edge[] node [left] {} (012)
    edge[] node [left] {} (013)
    edge[] node [left] {} (014)
(02) edge[] node [left] {} (012)
    edge[] node [left] {} (023)
    edge[] node [left] {} (024)
(03) edge[] node [left] {} (013)
    edge[] node [left] {} (023)
    edge[] node [left] {} (034)
(04) edge[] node [left] {} (014)
    edge[] node [left] {} (024)
    edge[] node [left] {} (034)
(012) edge[] node [left] {} (0123)
    edge[] node [left] {} (0124)
(013) edge[] node [left] {} (0123)
    edge[] node [left] {} (0134)
(023) edge[] node [left] {} (0123)
    edge[] node [left] {} (0234)
(014) edge[] node [left] {} (0124)
    edge[] node [left] {} (0134)
(024) edge[] node [left] {} (0124)
    edge[] node [left] {} (0234)
(034) edge[] node [left] {} (0134)
    edge[] node [left] {} (0234)
(0123) edge[] node [left] {} (01234)
(0124) edge[] node [left] {} (01234)
(0134) edge[] node [left] {} (01234)
(0234) edge[] node [left] {} (01234)
;
\end{tikzpicture}
\caption{\blue{An illustration of the set of all TA reductions from  Example~\ref{ex:msr2}. We denote individual reductions of $\TA$ using a bit-vector representation; for instance, $1101$ represents the reduction $\TA_{<D,I>}$ where $D \cup I = \{c_1, c_2, c_4\}$.
The reductions with a red dashed border are the insufficient reductions, and the reductions with solid green border are sufficient reductions. The MRSes and MIRs are filled with a background color.}}%
\label{fig:ex-powerset}
\end{figure}

\subsection{Problem Formulations}\label{sec:problem-formulation}
In this paper, we are mainly concerned with the following two problems. \blue{The first problem and the proposed solution were presented in our conference paper~\cite{DBLP:conf/tacas/BendikSGC21}.}

\begin{prob}\label{prob:msr-relax}
Given a TA $\TA = (L, l_0, C, \Delta, Inv )$ and a set of target locations $L_T \subset L$ that is unreachable on $\TA$,
find a \emph{minimal} ($D,I,\mathbf{r}$)-relaxation $\TA_{<D,I,\mathbf{r}>}$ of $\TA$ such that $L_T$ is reachable on $\TA_{<D,I,\mathbf{r}>}$. In particular, the goal is to identify a ($D,I,\mathbf{r}$)-relaxation that minimizes the number $|D \cup I|$ of relaxed constraints, and, secondly, we tend to minimize
the overall change of the clock constraints $\sum_{c \in D \cup I} \mathbf{r}(c)$.
\end{prob}

Our solution to Problem~\ref{prob:msr-relax} is described in detail in Sections~\ref{sec:msr} and~\ref{sec:msr-relax}.
Briefly, we
solve Problem~\ref{prob:msr-relax} in two steps.
First,  we identify a minimum MSR $\TA_{<D,I>}$ for $\mathcal{A}$, i.e., a minimal set $D \cup I$ of simple clock constraints whose removal from $\TA$ makes the target locations $L_T$ reachable. Second, instead of completely removing the constraints, we
turn the MSR $\TA_{<D,I>}$ into the resultant ($D,I,\mathbf{r}$)-relaxation. To construct the ($D,I,\mathbf{r}$)-relaxation, we propose two alternative approaches: (1) an approach based on Mixed Integer Linear Programming (MILP) and (2) an approach based on parameter synthesis for PTA\@.

\begin{prob}\label{prob:mg-relax}
Given a TA $\TA = (L, l_0, C, \Delta, Inv )$ and a set of target locations $L_T \subset L$ that is unreachable on $\TA$,
find a \emph{maximal} ($D,I,\mathbf{r}$)-relaxation $\TA_{<D,I,\mathbf{r}>}$ of $\TA$ such that $L_T$ is still unreachable on $\TA_{<D,I,\mathbf{r}>}$.
In particular, the goal is to identify a ($D,I,\mathbf{r}$)-relaxation that maximizes the number
$|\{ c \in D \cup I\, |\, \mathbf{r}(c) = \infty \}|$
of constraints that are completely removed, and, secondary, maximizes
the overall change of the clock constraints $\sum_{c \in \{ c' \in D \cup I\, |\, \mathbf{r}(c') \neq \infty \}} \mathbf{r}(c)$ that are not completely removed.
\end{prob}

Our solution to Problem~\ref{prob:mg-relax} is presented in Sections~\ref{sec:mg} and~\ref{sec:mg-relax}. Briefly, we first identify a minimum MG $D' \cup I'$ for $\mathcal{A}$, i.e., a minimal set $D' \cup I'$ of simple clock constraints that need to be left in $\TA$ to ensure that the target location $L_T$ is still unreachable. Subsequently, we employ parameter synthesis to further relax (as much as possible) the constraints $D' \cup I'$ that are left in the system. \blue{For both of the considered problems, we assume that there is a path from the initial state to the target set $L_T$ (unrealizable since $L_T$ is not reachable). Thus, the target set can become reachable via constraint removals/relaxations. }

\section{Finding Minimal Sufficient Reductions}\label{sec:msr}

\begin{algorithm}[t!]
\DontPrintSemicolon

$N \gets \TA_{<\Psi(\Delta),\Psi(Inv)>};\, \mathcal{I} \gets \emptyset;\, \mathcal{S} \gets \emptyset $\;
\While{$N \neq \mathtt{null}$}{
	$M, \mathcal{I} \gets \shrink(N,\mathcal{I})$ \tcp*{Algorithm~\ref{alg:single-msr-extraction}}
	$\mathcal{S} \gets \mathcal{S} \cup \{ X\, |\, M \sqsubseteq X\}$\;
	$\mathcal{I} \gets \mathcal{I} \cup \{ Y\, |\, Y \sqsubsetneq M\}$\;
	$N, \mathcal{I}, \mathcal{S} \gets \findSSeed(M, \mathcal{M}_{msr}, \mathcal{I}, \mathcal{S})$ \tcp*{Algorithm~\ref{alg:find-seed}}
}
\Return $M$\;

\caption{Minimum MSR Extraction Scheme}%
\label{alg:msr-scheme}
\end{algorithm}

In this section, we gradually describe our approach for finding a minimum minimal sufficient reduction.

\subsection{Base Scheme For Computing a Minimum MSR}

Algorithm~\ref{alg:msr-scheme} shows a high-level scheme of our approach for computing a minimum MSR\@.
The algorithm iteratively identifies an ordered set of MSRs, $|M_1| > |M_2| > \cdots > |M_k|$, such that the last MSR $M_k$ is a minimum MSR\@.
Each of the MSRs, say $M_i$, is identified in two steps. First, the algorithm finds an \emph{s-seed}\footnote{\blue{Note that the initial ``s'' in ``s-seed'' stands for ``sufficient''. Later, in Section~\ref{sec:mg}, we dually introduce ``i-seeds'' as ``insufficient'' reductions.}},
i.e., a reduction $N_i$ such that $N_i$ is sufficient and $|N_i| < |M_{i-1}|$. Second, the algorithm \emph{shrinks} $N_i$ into an MSR $M_i$ such that $M_i \sqsubseteq N_i$ (and thus $|M_i| \leq |N_i|$). The initial s-seed $N_1$ is $\TA_{<\Psi(\Delta),\Psi(Inv)>}$, i.e., the reduction that removes all simple clock constraints (which makes all locations of $\TA$ trivially reachable).
Once there is no sufficient reduction $N_i$ with $|N_i| < |M_{i-1}|$, we know that $M_{i - 1} = M_{k}$ is a minimum MSR\@.

Note that the algorithm also maintains two auxiliary sets, $\mathcal{I}$ and $\mathcal{S}$, to store all identified insufficient and sufficient reductions, respectively. In particular, whenever we identify a new MSR $M_i$, we add every reduction $X$ such that $M_i \sqsubseteq X$ to $\mathcal{S}$ since, by Proposition~\ref{proposition:monotonicity}, every such $X$ is sufficient. Dually, since $M_i$ is an MSR, then every reduction $Y$ such that $Y \sqsubsetneq M_i$ is necessarily insufficient, and hence we add it to $\mathcal{I}$.
The sets  $\mathcal{I}$ and $\mathcal{S}$ are used during the process of finding and shrinking an s-seed which we describe below.

\subsection{Shrinking an S-Seed}
Our approach for shrinking an s-seed $N$ into an MSR $M$ is based on two concepts: a \emph{critical simple clock constraint} and a \emph{reduction core}.

\begin{defi}[critical constraint]\label{defn:critical-constraint}
Given a sufficient reduction $\TA_{<D,I>}$, a simple clock constraint $c$ is \emph{critical} for $\TA_{<D,I>}$ iff $\TA_{<D\setminus \{c\},I\setminus \{c\}>}$ is insufficient.
\end{defi}

\begin{prop}\label{prop:critical-constraint}
If $c \in D \cup I$ is critical for a sufficient reduction  $\TA_{<D,I>}$
then $c$ is critical for every sufficient reduction $\TA_{<D',I'>}$ such that $\TA_{<D',I'>} \sqsubseteq \TA_{<D,I>}$. Moreover,
 by Definitions~\ref{defn:msr} and~\ref{defn:critical-constraint}, $\TA_{<D,I>}$ is an MSR iff every $c \in D \cup I$ is \emph{critical} for~$\TA_{<D,I>}$.
\end{prop}

\begin{proof}
By contradiction, assume that $c$ is critical for $\TA_{<D,I>}$ but not for $\TA_{<D',I'>}$,
 i.e., $\TA_{<D\setminus \{c\},I \setminus \{c\}>}$ is insufficient and $\TA_{<D' \setminus \{c\},I' \setminus \{c\}>}$ is sufficient.
As $\TA_{<D',I'>} \sqsubseteq \TA_{<D,I>}$, we have $\TA_{<D' \setminus \{c\} ,I' \setminus \{c\}>} \sqsubseteq \TA_{<D \setminus \{c\},I \setminus \{c\}>}$. By Proposition~\ref{proposition:monotonicity}, if the reduction $\TA_{<D' \setminus \{c\},I' \setminus \{c\}>}$ is sufficient then $\TA_{<D\setminus \{c\},I \setminus \{c\}>}$ is also sufficient.
\end{proof}

\begin{algorithm}[t!]
\DontPrintSemicolon

$X \gets \emptyset$\;
\While{$(D \cup I) \neq X$}{
	$c \gets $ pick a simple clock constraint from $(D \cup I) \setminus X$\;
	\eIf{$\TA_{<D \setminus \{c\},I \setminus \{c\}>} \not\in \mathcal{I}$ \textbf{and} $\TA_{<D \setminus \{c\},I \setminus \{c\}>}$ is sufficient }{\label{alg:lazy_if}
		$\rho \gets $ a witness run of the sufficiency of $\TA_{<D \setminus \{c\},I \setminus \{c\}>}$\;
		$\TA_{<D,I>} \gets $ the reduction core of $\TA_{<D \setminus \{c\},I \setminus \{c\}>}$ w.r.t. $\rho$\;
	}{
		$X \gets X \cup \{c\}$\;
		$\mathcal{I} \gets \mathcal{I} \cup \{N \in \mathcal{R}_{\TA}\, |\, N \sqsubseteq \TA_{<D \setminus \{c\},I \setminus \{c\}>} \}$	\;
	}
}
\Return $\TA_{<D,I>},\mathcal{I}$\;

\caption{$\shrink{(\TA_{<D,I>},\mathcal{I})}$}%
\label{alg:single-msr-extraction}
\end{algorithm}

\begin{defi}[reduction core]\label{defn:reduction-core}
Let $\TA_{<D,I>}$ be a sufficient reduction,
$\rho$ a witness run of the sufficiency (i.e., reachability of $L_T$ on $\TA_{<D,I>}$), and $\pi$ the path corresponding to $\rho$.
Furthermore, let $\pi' = l_0,e_1,\ldots,e_{n},l_n$ be the path corresponding to $\pi$ on the original TA $\TA$ \blue{(i.e., $\pi = M(\pi')$~\eqref{eq:path_map})}.
The \emph{reduction core} of $\TA_{<D,I>}$ w.r.t. $\rho$ is the reduction $A_{<D',I'>}$ where $D' = \{ (e,\varphi)\, |\, (e,\varphi) \in D \wedge e = e_i$ for some $1 \leq i \leq n\}$ and $I' = \{ (l,\varphi)\, |\, (l,\varphi) \in I \wedge l = l_i$ for some $0 \leq l \leq n\}$.
\end{defi}

Intuitively, the reduction core of $\TA_{<D,I>}$ w.r.t. $\rho$ removes from $\TA$ only the simple clock constraints that appear on the witness path.

\begin{prop}
Let $\TA_{<D,I>}$ be a sufficient reduction, $\rho$ the witness of reachability of $L_T$ on $\TA_{<D,I>}$, and $\TA_{<D',I'>}$ the reduction core of $\TA_{<D,I>}$ w.r.t. $\rho$. Then $\TA_{<D',I'>}$ is a sufficient reduction and $\TA_{<D',I'>} \sqsubseteq \TA_{<D,I>}$.
\end{prop}

\begin{proof}
By Definition~\ref{defn:reduction-core}, $D' \subseteq D$ and $I' \subseteq I$, thus $\TA_{<D',I'>} \sqsubseteq \TA_{<D,I>}$.
As for the sufficiency of $\TA_{<D',I'>}$, we only sketch the proof.
Intuitively, both $\TA_{<D,I>}$ and $\TA_{<D',I'>}$ originate from $\TA$ by only removing some simple clock constraints ($D \cup I$, and $D' \cup I'$, respectively), i.e., the graph structure of $\TA_{<D,I>}$ and $\TA_{<D',I'>}$ is the same, however, some corresponding paths of $\TA_{<D,I>}$ and $\TA_{<D',I'>}$ differ in the constraints that appear on the paths. By Definition~\ref{defn:reduction-core}, the path $\pi$ that corresponds to the witness run $\rho$ of $\TA_{<D,I>}$ is also a path of $\TA_{<D',I'>}$. Since realizability of a path depends only on the constraints along  the path, if $\pi$ is realizable on $\TA_{<D,I>}$ then $\pi$ is also realizable on $\TA_{<D',I'>}$.
\end{proof}

Our approach for shrinking a sufficient reduction $N$
is shown in Algorithm~\ref{alg:single-msr-extraction}. The algorithm iteratively maintains a sufficient reduction $\TA_{<D,I>}$ and a set $X$ of known critical constraints for $\TA_{<D,I>}$. Initially, $\TA_{<D,I>} = N$ and $X = \emptyset$. In each iteration, the algorithm picks a simple clock constraint $c \in (D \cup I) \setminus X$ and checks the reduction $\TA_{<D \setminus \{c\},I \setminus \{c\}>}$ for sufficiency. If $\TA_{<D \setminus \{c\},I \setminus \{c\}>}$ is insufficient, the algorithm adds $c$ to $X$. Otherwise, if $\TA_{<D \setminus \{c\},I \setminus \{c\}>}$ is sufficient, the algorithm obtains a witness run $\rho$ of the sufficiency from the verifier and reduces $\TA_{<D, I>}$ to the corresponding reduction core. The algorithm terminates when $(D \cup I) = X$. An invariant of the algorithm is that
every $c \in X$ is critical for $\TA_{<D,I>}$.
Thus, when $(D \cup I) = X$,  $\TA_{<D, I>}$ is an MSR (Proposition~\ref{prop:critical-constraint}).

Note that the algorithm also uses the set $\mathcal{I}$ of known insufficient reductions.
In particular, before calling a verifier to check a reduction for sufficiency (line~\ref{alg:lazy_if}), the algorithm first checks (in a lazy manner) whether the reduction is already known to be insufficient. Also, whenever the algorithm determines a reduction $\TA_{<D \setminus \{c\},I \setminus \{c\}>}$ to be insufficient, it adds $\TA_{<D \setminus \{c\},I \setminus \{c\}>}$ and every $N$, $N \sqsubseteq \TA_{<D \setminus \{c\},I \setminus \{c\}>}$, to $\mathcal{I}$ (by Proposition~\ref{proposition:monotonicity}, every such $N$ is also insufficient).

Finally, note that the algorithm does not add any reduction to the set $\mathcal{S}$ even though it can identify some sufficient reductions during its computation. The reason is that every such identified reduction is larger (w.r.t. $\sqsubseteq$) than the resultant MSR, and hence all these sufficient reductions are added to $\mathcal{S}$ in the main procedure (Algorithm~\ref{alg:msr-scheme}) after the shrinking.

\subsection{Finding an S-Seed}
We now describe the procedure {\findSSeed} that, given the latest identified MSR $M$, identifies an s-seed, i.e., a sufficient
reduction $N$ such that $|N| < |M|$, or returns \texttt{null} if there is no s-seed. Let us denote by \texttt{CAND} the set of all \emph{candidates} on an s-seed, i.e., $\mathtt{CAND} = \{N \in \mathcal{R}_{\TA}\, |\, |N| < |M| \}$.
A brute-force approach would be to check individual reductions in $\mathtt{CAND}$ for sufficiency until a sufficient one is found, however, this can be practically intractable since $|\mathtt{CAND}| = \sum_{i=1}^{|M|}\binom{|\Psi(\Delta) \cup \Psi(Inv)|}{i - 1}$.

We provide two observations to prune the set \texttt{CAND} of candidates that need to be tested for being an s-seed. The first observation exploits the set $\mathcal{I}$ of already known insufficient reductions: no $N \in \mathcal{I}$ can be an s-seed. The second observation is stated below:

\begin{obs}\label{obs:sseed-max}
For every sufficient reduction $N \in \mathtt{CAND}$ there exists a sufficient reduction $N' \in \mathtt{CAND}$ such that $N \sqsubseteq N'$ and $|N'| = |M|-1$.
\end{obs}

\begin{proof}
If $|N| = |M|-1$, then $N = N'$. For the other case, when $|N| < |M|-1$, let $N = \TA_{<D^N,I^N>}$ and $M = \TA_{<D^M,I^M>}$.
We construct $N' = \TA_{<D^{N'},I^{N'}>}$ by adding arbitrary $(|M| - |N|) - 1$ simple clock constraint from $(D^M \cup I^M) \setminus (D^N \cup I^N)$ to $(D^N \cup I^N)$, i.e., $D^N \cup I^N \subseteq D^{N'} \cup I^{N'} \subseteq (D^M \cup I^M \cup D^N \cup I^N)$ and $|D^{N'} \cup I^{N'}| = |M|-1$.
By definition of \texttt{CAND}, $N' \in \mathtt{CAND}$. Moreover, since $N \sqsubsetneq N'$ and $N$ is sufficient, then $N'$ is also sufficient (Proposition~\ref{proposition:monotonicity}).
\end{proof}

\begin{algorithm}[t]
\DontPrintSemicolon

\While{$ \{ N \in \mathcal{R}_{\TA}\, |\, N \not\in \mathcal{I} \wedge  |N|= |M| -1\} \neq \emptyset$}{
	$N \gets $ pick from $\{ N \in \mathcal{R}_{\TA}\, |\, N \not\in \mathcal{I} \wedge |N|= |M| -1\}$\;
	\lIf{$N$ is sufficient}{
		\Return $N,\mathcal{I},\mathcal{S}$
	}
	\Else
	{
			$E, \mathcal{S} \gets \enlarge(N, \mathcal{S})$ \tcp*{Algorithm~\ref{alg:single-mir-extraction}}
			$\mathcal{I} \gets \mathcal{I} \cup \{N' \in \mathcal{R}_{\TA}\, |\, N' \sqsubseteq E \}$\;
	}
}
\Return $\mathtt{null},\mathcal{I}, \mathcal{S}$

\caption{$\findSSeed(M, \mathcal{I}, \mathcal{S})$}%
\label{alg:find-seed}
\end{algorithm}

Based on the above observations, we build a set $\mathcal{C}_s$ of indispensable candidates on s-seeds that need to be tested for sufficiency:
\begin{equation}\label{eq:seed_candidates}
\mathcal{C}_s = \{ N \in \mathcal{R}_{\TA}\, |\, N \not\in \mathcal{I} \wedge |N|= |M| -1\}
\end{equation}

The procedure {\findSSeed}, shown in Algorithm~\ref{alg:find-seed}, in each iteration picks a reduction $N \in \mathcal{C}_s$ and checks it for sufficiency (via the verifier). If $N$ is sufficient, {\findSSeed} returns $N$ as the s-seed.
Otherwise, when $N$ is insufficient, the algorithm first \emph{enlarges} $N$ into a maximal insufficient reduction (MIR) $E$ such that $N \sqsubseteq E$. By Proposition~\ref{proposition:monotonicity}, every reduction $N'$ such that $N' \sqsubseteq E$ is also insufficient, thus all these reductions are subsequently added to $\mathcal{I}$ and hence  removed from $\mathcal{C}_s$ (note that this includes also $N$).
If $\mathcal{C}_s$ becomes empty, then there is no  s-seed.

The purpose of \emph{enlarging} $N$ into $E$ is to quickly prune the candidate set $\mathcal{C}_s$. We could just add all the insufficient reductions $\{N'\, |\, N' \sqsubseteq N\}$ to $\mathcal{I}$, but note that
$|\{N'\, |\, N' \sqsubseteq E\}|$ is exponentially larger than $|\{N'\, |\, N' \sqsubseteq N\}|$ w.r.t. $|E| - |N|$.
The enlargement of $N$ into an MIR $E$ is carried out via Algorithm~\ref{alg:single-mir-extraction} and it is described later on in Section~\ref{sec:mg}. Note that Algorithm~\ref{alg:single-mir-extraction} exploits and updates the set $\mathcal{S}$ of already known sufficient reductions.

Finally, let us note that we need to somehow efficiently represent and maintain the sets $\mathcal{I}$, $\mathcal{S}$ and $\mathcal{C}_s$. In particular, we need to be able to add elements to these sets and obtain elements from these sets. The problem is that there can be up to exponentially many reductions w.r.t. $|\Psi(\Delta) \cup \Psi(Inv)|$, and hence these sets can be also exponentially large and cannot be stored explicitly. In Section~\ref{sec:representation}, we describe how we efficiently maintain these sets.

\subsection{Example Execution}

\tikzstyle{uvalid}=[draw,dashed]
\tikzstyle{valid}=[draw=red,dashed,fill=red!20]
\tikzstyle{uinvalid}=[draw]
\tikzstyle{invalid}=[draw=green,fill=green!20]
\tikzstyle{maximal}=[draw=blue,fill=blue!20]

\begin{figure}[ht]
\begin{minipage}{0.49\textwidth}
\centering
\begin{tikzpicture}[-,>=stealth',shorten >=1pt,auto,node distance=1.5cm,minimum size=0.4cm,
thick,every node/.style={draw, ellipse,inner sep=0, outer sep=0},
mus/.style={draw,draw=green,fill=green!20,dashed},
mss/.style={draw,draw=red,fill=red!20,dashed},
v/.style={fill=green},
i/.style={fill=red}]
\node[valid] (0) []{\tiny 0000};
\node[uvalid,maximal] (01) [above left = 0.44999999999999996 and -2.11 of 0]{\tiny 1000};
\node[uvalid,maximal] (02) [above left = 0.44999999999999996 and -1.005 of 0]{\tiny 0100};
\node[valid] (03) [above left = 0.44999999999999996 and 0.1 of 0]{\tiny 0010};
\node[valid] (04) [above left = 0.44999999999999996 and 1.205 of 0]{\tiny 0001};

\node[uvalid] (012) [above left = 1.15 and -2.9999999999999996 of 0]{\tiny 1100};
\node[uvalid] (013) [above left = 1.15 and -1.9499999999999998 of 0]{\tiny 1010};
\node[uvalid] (023) [above left = 1.15 and -1 of 0]{\tiny 0110};
\node[uvalid] (014) [above left = 1.15 and 0.05 of 0]{\tiny 1001};
\node[uvalid] (024) [above left = 1.15 and 0.9999999999999999 of 0]{\tiny 0101};
\node[mus] (034) [above left = 1.15 and 2.0500000000000001 of 0]{\tiny 0011};

\node[uinvalid] (0123) [above left = 1.8499999999999996 and -2.11 of 0]{\tiny 1110};
\node[uvalid] (0124) [above left = 1.8499999999999996 and -1.005 of 0]{\tiny 1101};
\node[invalid] (0134) [above left = 1.8499999999999996 and 0.1 of 0]{\tiny 1011};
\node[invalid] (0234) [above left = 1.8499999999999996 and 1.205 of 0]{\tiny 0111};
\node[invalid] (01234) [above left = 2.55 and -0.475 of 0]{\tiny 1111};
\path[every node/.style={font=\sffamily\small}]
(0) edge[] node [left] {} (01)
    edge[] node [left] {} (02)
    edge[] node [left] {} (03)
    edge[] node [left] {} (04)
(01) edge[] node [left] {} (012)
    edge[] node [left] {} (013)
    edge[] node [left] {} (014)
(02) edge[] node [left] {} (012)
    edge[] node [left] {} (023)
    edge[] node [left] {} (024)
(03) edge[] node [left] {} (013)
    edge[] node [left] {} (023)
    edge[] node [left] {} (034)
(04) edge[] node [left] {} (014)
    edge[] node [left] {} (024)
    edge[] node [left] {} (034)
(012) edge[] node [left] {} (0123)
    edge[] node [left] {} (0124)
(013) edge[] node [left] {} (0123)
    edge[] node [left] {} (0134)
(023) edge[] node [left] {} (0123)
    edge[] node [left] {} (0234)
(014) edge[] node [left] {} (0124)
    edge[] node [left] {} (0134)
(024) edge[] node [left] {} (0124)
    edge[] node [left] {} (0234)
(034) edge[] node [left] {} (0134)
    edge[] node [left] {} (0234)
(0123) edge[] node [left] {} (01234)
(0124) edge[] node [left] {} (01234)
(0134) edge[] node [left] {} (01234)
(0234) edge[] node [left] {} (01234)
;
\end{tikzpicture}
\caption{The situation before the first call of $\findSSeed$.}%
\label{ex:execution2}
\end{minipage}%
\hfill
\begin{minipage}{0.49\textwidth}
\centering
\begin{tikzpicture}[-,>=stealth',shorten >=1pt,auto,node distance=1.5cm,minimum size=0.4cm,
thick,every node/.style={draw, ellipse,inner sep=0, outer sep=0},
mus/.style={draw,draw=green,fill=green!20,dashed},
mss/.style={draw,draw=red,fill=red!20,dashed},
v/.style={fill=green},
i/.style={fill=red}]
\node[valid] (0) []{\tiny 0000};
\node[valid] (01) [above left = 0.44999999999999996 and -2.11 of 0]{\tiny 1000};
\node[valid] (02) [above left = 0.44999999999999996 and -1.005 of 0]{\tiny 0100};
\node[valid] (03) [above left = 0.44999999999999996 and 0.1 of 0]{\tiny 0010};
\node[valid] (04) [above left = 0.44999999999999996 and 1.205 of 0]{\tiny 0001};

\node[valid] (012) [above left = 1.15 and -2.9999999999999996 of 0]{\tiny 1100};
\node[uvalid] (013) [above left = 1.15 and -1.9499999999999998 of 0]{\tiny 1010};
\node[uvalid] (023) [above left = 1.15 and -1 of 0]{\tiny 0110};
\node[valid] (014) [above left = 1.15 and 0.05 of 0]{\tiny 1001};
\node[valid] (024) [above left = 1.15 and 0.9999999999999999 of 0]{\tiny 0101};
\node[mus] (034) [above left = 1.15 and 2.0500000000000001 of 0]{\tiny 0011};

\node[uinvalid] (0123) [above left = 1.8499999999999996 and -2.11 of 0]{\tiny 1110};
\node[valid] (0124) [above left = 1.8499999999999996 and -1.005 of 0]{\tiny 1101};
\node[invalid] (0134) [above left = 1.8499999999999996 and 0.1 of 0]{\tiny 1011};
\node[invalid] (0234) [above left = 1.8499999999999996 and 1.205 of 0]{\tiny 0111};
\node[invalid] (01234) [above left = 2.55 and -0.475 of 0]{\tiny 1111};
\path[every node/.style={font=\sffamily\small}]
(0) edge[] node [left] {} (01)
    edge[] node [left] {} (02)
    edge[] node [left] {} (03)
    edge[] node [left] {} (04)
(01) edge[] node [left] {} (012)
    edge[] node [left] {} (013)
    edge[] node [left] {} (014)
(02) edge[] node [left] {} (012)
    edge[] node [left] {} (023)
    edge[] node [left] {} (024)
(03) edge[] node [left] {} (013)
    edge[] node [left] {} (023)
    edge[] node [left] {} (034)
(04) edge[] node [left] {} (014)
    edge[] node [left] {} (024)
    edge[] node [left] {} (034)
(012) edge[] node [left] {} (0123)
    edge[] node [left] {} (0124)
(013) edge[] node [left] {} (0123)
    edge[] node [left] {} (0134)
(023) edge[] node [left] {} (0123)
    edge[] node [left] {} (0234)
(014) edge[] node [left] {} (0124)
    edge[] node [left] {} (0134)
(024) edge[] node [left] {} (0124)
    edge[] node [left] {} (0234)
(034) edge[] node [left] {} (0134)
    edge[] node [left] {} (0234)
(0123) edge[] node [left] {} (01234)
(0124) edge[] node [left] {} (01234)
(0134) edge[] node [left] {} (01234)
(0234) edge[] node [left] {} (01234)
;
\end{tikzpicture}
\caption{The situation after the first call of $\findSSeed$.}%
\label{ex:execution3}
\end{minipage}%
\end{figure}

\blue{
We illustrate an execution of Algorithm~\ref{alg:msr-scheme} on the TA $\TA$ defined in Example~\ref{ex:running} (Fig.~\ref{fig:ex}) with an initial location $l_0$ and a target unreachable set of locations $L_T = \{ l_4\}$.
For the sake of a graphical illustration, we restrict our analysis to possible removal of only 4 simple clock constraints: $c_1 = x \geq 9$, $c_2 = z \geq 9$, $c_3 = x \leq 14$ and $c_4 = u \leq 26$ that appear on edge $e_1$, edge $e_5$, location $l_1$, and location $l_3$, respectively (same as in Example~\ref{ex:msr2}).
We will use a bitvector notation to denote the individual reductions, e.g., $\TA_{1011}$ represents the reduction $\TA_{<D,I>}$ where $D \cup I = \{ c_1, c_3, c_4\}$.

The computation starts by setting $N$ to $\TA_{1111}$, $\mathcal{I} = \emptyset$ and $\mathcal{S} = \emptyset$. Subsequently, in the first iteration of Algorithm~\ref{alg:msr-scheme}, $N$ is shrunk into an MSR $M$. Assume that $M = \TA_{0011}$, and that $\mathcal{I}$ was enlarged to $\mathcal{I} = \{ \TA_{0001},\TA_{0010},\TA_{0000}  \}$. After the shrinking, Algorithm~\ref{alg:msr-scheme} also enlarges the sets $\mathcal{I}$ and $\mathcal{S}$ by adding to them reductions that are smaller and larger than $M$ w.r.t. $\sqsubseteq$ and $\sqsupseteq$, respectively.
We depict the situation at this moment in Figure~\ref{ex:execution2}. The power-set in the figure represents all possible reductions of $\TA$ (in the picture, we denote a reduction $\TA_B$ by the bitvector $B$). The reductions with dashed border are insufficient, and the reductions with solid border are sufficient. We use green and red background color to highlight the reductions in sets $\mathcal{S}$ and $\mathcal{I}$, respectively. Moreover, we highlight in blue two reductions, $\TA_{0100}$ and $\TA_{1000}$, that will form the set $\mathcal{C}_s$ in the subsequent call of $\findSSeed(M, \mathcal{I}, \mathcal{S})$.

During the execution of $\findSSeed(M, \mathcal{I}, \mathcal{S})$, assume we first pick the candidate reduction $N = \TA_{0100} \in \mathcal{C}_s$ and check it for sufficiency. It is insufficient, hence we enlarge it (via $\enlarge(N, \mathcal{S})$) to an insufficient reduction $E$; assume $E = \TA_{1101}$. Subsequently, we add to $\mathcal{I}$ every reduction $N'$ such that $N' \sqsubseteq E$. The situation at this moment is depicted in Figure~\ref{ex:execution3}. At this point, $\mathcal{C}_s$ is empty, i.e., we have the guarantee that there is no s-seed that would be smaller than $M$ w.r.t. $\sqsubsetneq$. Hence, $\findSSeed$ terminates, and Algorithm~\ref{alg:msr-scheme} then also terminates determining that the $M$ from the first (and only) iteration is a minimum MSR\@.

Finally, let us note that there are different possible executions of our algorithm on the given example. In particular, in Algorithm~\ref{alg:find-seed}, we choose a reduction $N$ from the candidate set $\mathcal{C}_s$ and the choice determines which sufficient reduction will be produced (if any). Similarly, in Algorithm~\ref{alg:single-msr-extraction}, we pick constraints $c$ in some order and this order determines which MSR will be produced. We observed that different reduction and constraint choices affect the performance of the overall algorithm, both in the runtime and the number of performed verifier calls. However, we postpone a development of a suitable heuristic for making good choices here for a future work.
}

\section{Finding Maximal Insufficient Reductions}\label{sec:mg}

In this section, we describe our approach for finding maximum maximal insufficient reductions (MIRs), and consequently also their complementary  minimum minimal guarantees (MGs).

\subsection{Base scheme for Computing a Maximum MIR}
Our scheme for computing a maximum MIR is shown in Algorithm~\ref{alg:mir-scheme} and it works in a dual way to the scheme for computing a minimum MSR (Algorithm~\ref{alg:msr-scheme}).
We iteratively identify a sequence $M_1, M_2, \ldots, M_k$ of MIRs such that $|M_1| < |M_2| < \cdots < |M_k|$ and the last MIR, $M_k$, is a maximum MIR\@.
To find each MIR $M_i$ in the sequence, we proceed in two steps. First, we identify an \emph{i-seed}, i.e., an insufficient reduction $N_i$ such that $|N_i| > |M_{i-1}|$. Second, we \emph{enlarge} $N_i$ into the MIR $M_i$ (i.e., $N_i \sqsubseteq M_i$ and hence $|N_i| \leq |M_i|$). Once there is no more i-seed, it is guaranteed that the last identified MIR $M_{i-1} = M_k$ is a maximum MIR\@.
The initial i-seed $N_1$ is the reduction $\TA_{<\emptyset,\emptyset>} = \TA$ (we assume that $L_T$ is indeed unreachable on the input TA $\TA$).

Same as in case of Algorithm~\ref{alg:msr-scheme}, this scheme also maintains the auxiliary sets $\mathcal{I}$ and $\mathcal{S}$ to store all identified insufficient and sufficient reductions, respectively.

\begin{algorithm}[t!]
\DontPrintSemicolon

$N \gets \TA_{<\emptyset,\emptyset>};\, \mathcal{I} \gets \emptyset;\, \mathcal{S} \gets \emptyset $\;
\While{$N \neq \mathtt{null}$}{
	$M, \mathcal{S} \gets \enlarge(N,\mathcal{S})$ \tcp*{Algorithm~\ref{alg:single-mir-extraction}}
	$\mathcal{I} \gets \mathcal{I} \cup \{ M'\, |\, M' \sqsubseteq M\}$\;\label{alg-line:ins-block}
	$N, \mathcal{S} \gets \findISeed(M, \mathcal{I}, \mathcal{S})$ \tcp*{Algorithm~\ref{alg:find-iseed}}
}
\Return $M$\;

\caption{Maximum MIR Extraction Scheme}%
\label{alg:mir-scheme}
\end{algorithm}

\begin{algorithm}[t!]
\DontPrintSemicolon

$X \gets \emptyset$\;
\While{$(\Psi(\Delta) \cup \Psi(Inv)) \setminus (D \cup I ) \neq X$}{
	$c \gets $ pick a constraint from $(\Psi(\Delta) \cup \Psi(Inv)) \setminus (D \cup I \cup X)$\;
	let $R = \TA_{<D \cup \{c\},I>}$ if $c \in \Psi(\Delta)$ and $R = \TA_{<D,I \cup \{c\}>}$ otherwise\;
	\eIf{$R \not\in \mathcal{S}$ \textbf{and} $R$ is not sufficient}{
		$\TA_{<D,I>} \gets R$\;
	}{
		$X \gets X \cup \{c\}$\;
		$\mathcal{S} \gets \mathcal{S} \cup \{N \in \mathcal{R}_{\TA}\, |\, N \sqsupseteq R \}$	\;
	}

%
}
\Return $\TA_{<D,I>},\mathcal{S}$\;

\caption{$\enlarge(\TA_{<D,I>},\mathcal{S})$}%
\label{alg:single-mir-extraction}
\end{algorithm}

\subsection{Enlarging an I-Seed}
The procedure {\enlarge} is based on a concept of \emph{conflicting simple clock constraints}.

\begin{defi}[conflicting constraint]\label{defn:conflicting-constraint}
Given an insufficient reduction $\TA_{<D,I>}$, a simple clock constraint $c \in (\Psi(\Delta) \cup \Psi(Inv)) \setminus (D \cup I )$ is \emph{conflicting} for $\TA_{<D,I>}$ if the reduction  $\TA_{<D',I'>}$ with $D' \cup I' = D \cup I \cup \{c\}$ is sufficient.
\end{defi}

Note that if a constraint $c$ is conflicting for an insufficient reduction $N$ then $c$ is also conflicting for every insufficient reduction $N'$ with $N \sqsubseteq N'$. Moreover, note that a reduction $\TA_{<D,I>}$ is an MIR iff every $c \in (\Psi(\Delta) \cup \Psi(Inv)) \setminus (D \cup I )$ is conflicting for $\TA_{<D,I>}$.

The procedure $\enlarge(N, \mathcal{S})$ is shown in Algorithm~\ref{alg:single-mir-extraction}.
The algorithm iteratively maintains an insufficient reduction $\TA_{<D,I>}$ and a set $X$ of constraints that are known to be conflicting for
 $\TA_{<D,I>}$. Initially, $\TA_{<D,I>} = N$ and $X = \emptyset$. In each iteration, the algorithm picks a simple clock constraint
 $c \in (\Psi(\Delta) \cup \Psi(Inv)) \setminus (D \cup I \cup X)$ and checks whether $c$ is conflicting for $\TA_{<D,I>}$.
 If $c$ is conflicting, then it is added to $X$. Otherwise, if $c$ is not conflicting, then $\TA_{<D,I>}$ is extended either to
 $\TA_{<D \cup \{c\},I>}$ or to  $\TA_{<D,I \cup \{c\}>}$ (depending on if $c \in \Psi(\Delta)$ or $c \in \Psi(Inv)$).
The algorithm maintains the invariant that
every $c \in X$ is conflicting for $\TA_{<D,I>}$, and hence when $X = (\Psi(\Delta) \cup \Psi(Inv)) \setminus (D \cup I )$, it is guaranteed that $\TA_{<D,I>}$ is an MIR\@.

The check whether a constraint  $c \in (\Psi(\Delta) \cup \Psi(Inv)) \setminus (D \cup I \cup X)$ is conflicting for $\TA_{<D,I>}$ is carried out by testing whether the reduction $\TA_{<D',I'>}$ with $D' \cup I' = D \cup I \cup \{c\}$ is sufficient. In particular, to save some invocations of the verifier, we first, in a lazy manner, check whether $\TA_{<D',I'>} \in \mathcal{S}$ (i.e., $\TA_{<D',I'>}$ is already known to be sufficient). If $\TA_{<D',I'>} \not\in \mathcal{S}$, we check $\TA_{<D',I'>}$ for sufficiency via the verifier.
Also, note that whenever we identify a sufficient reduction $\TA_{<D',I'>}$, we add every reduction $X$ such that $X \sqsupseteq \TA_{<D',I'>}$ to $\mathcal{S}$ (by Proposition~\ref{proposition:monotonicity}, every such $X$ is also sufficient).
Finally, note that we do not add any insufficient reduction that is identified during the enlargement to the set $\mathcal{I}$. The reason is that all insufficient reductions that are identified during the enlargement are smaller (w.r.t. $\sqsubseteq$) than the resultant MIR, and we update $\mathcal{I}$ based on the MIR after the enlargement (Algorithm~\ref{alg:mir-scheme}, line~\ref{alg-line:ins-block}).

\subsection{Finding an I-Seed}
\begin{algorithm}[t!]
\DontPrintSemicolon

\While{$ \{ N \in \mathcal{R}_{\TA}\, |\, N \not\in \mathcal{S} \wedge |N|= |M| + 1\} \neq \emptyset$}{
	$N \gets $ pick from $\{ N \in \mathcal{R}_{\TA}\, |\, N \not\in \mathcal{S} \wedge |N|= |M| +1\}$\;
	\lIf{$N$ is insufficient}{
		\Return $N,\mathcal{I},\mathcal{S}$
	}
	\Else
	{
			$E, \mathcal{I} \gets \shrink(N, \mathcal{I})$ \tcp*{Algorithm~\ref{alg:single-msr-extraction}}
			$\mathcal{S} \gets \mathcal{S} \cup \{N' \in \mathcal{R}_{\TA}\, |\, N' \sqsupseteq E \}$\;
	}
}
\Return $\mathtt{null},\mathcal{I}, \mathcal{S}$

\caption{$\findISeed(M, \mathcal{I}, \mathcal{S})$}%
\label{alg:find-iseed}
\end{algorithm}

The procedure \findISeed{} works dually to the procedure \findSSeed{}.
The input is the latest identified MIR $M$ and the sets $\mathcal{I}$ and $\mathcal{S}$ of known insufficient and sufficient reductions. The output is an i-seed, i.e., an insufficient
reduction $N$ such that $|N| > |M|$, or \texttt{null} if there is no i-seed.

We exploit two basic observations while searching for $N$. First, observe that no reduction that is already known to be sufficient, i.e., belongs to $\mathcal{S}$, can be an i-seed. Second, observe that:

\begin{obs}
For every insufficient reduction $N$ with $|N| > |M|$, there exists an insufficient $N'$ such that $N' \sqsubseteq N$ and $|N'| = |M| + 1$.
\end{obs}

\begin{proof}
Dually to the proof of Observation~\ref{obs:sseed-max}.
\end{proof}

Exploiting the above two observations, we build a set $\mathcal{C}_i$ of indispensable candidates on i-seeds that need to be tested for sufficiency to either find an i-seed or to prove that there are no more i-seeds:
\begin{equation}\label{eq:iseed_candidates}
\mathcal{C}_i = \{ N \in \mathcal{R}_{\TA}\, |\, N \not\in \mathcal{S} \wedge |N|= |M| + 1\}
\end{equation}

The procedure {\findISeed} (Algorithm~\ref{alg:find-iseed}) iteratively picks a reduction $N \in \mathcal{C}_i$ and checks it for sufficiency via the verifier. If $N$ is found to be insufficient, it is returned as the i-seed. Otherwise, when $N$ is sufficient,
the algorithm \emph{shrinks} $N$ to an MSR $E$ via Algorithm~\ref{alg:single-msr-extraction}.
By Proposition~\ref{proposition:monotonicity}, every reduction $N'$ such that $N' \sqsupseteq E$ is also sufficient; hence, we add all these reductions to $\mathcal{S}$ (and thus implicitly remove them from $\mathcal{C}_i$).
If $\mathcal{C}_i$ becomes empty, then there is no  i-seed.

\section{Representation of \texorpdfstring{$\mathcal{I}$}{I}, \texorpdfstring{$\mathcal{S}$}{S}, \texorpdfstring{$\mathcal{C}_s$}{Cs}, and \texorpdfstring{$\mathcal{C}_i$}{Ci}}\label{sec:representation}
Let us now describe how to efficiently represent and maintain the sets $\mathcal{I}$, $\mathcal{S}$, $\mathcal{C}_s$ and $\mathcal{C}_i$ that are used in our algorithms. Recall that we need to be able to add elements to these sets, obtain elements from these sets, and in case of $\mathcal{C}_s$ and $\mathcal{C}_i$ also perform emptiness checks.
The problem is that the size of these sets can be expontential w.r.t. $|\Psi(\Delta) \cup \Psi(Inv)|$ (there are exponentially many reductions), and thus, it is practically intractable to maintain the sets explicitly.
Instead, we use a symbolic representation.

Given a timed automaton $\TA$ with simple clock constraints $\Psi(\Delta) = \{ (e_1, \varphi_1), \ldots, (e_p, \varphi_p) \}$ and $\Psi(Inv) = \{ (l_1, \varphi_1), \ldots, (l_q, \varphi_q) \}$,
we introduce two sets of Boolean variables $X = \{x_1, \ldots, x_p\}$ and $Y = \{y_1, \ldots, y_q\}$. Note that every valuation of the variables $X \cup Y$ one-to-one maps to the reduction $\TA_{<D,I>}$ such that $(e_i, \varphi_i) \in D$ iff $x_i$ is assigned \emph{True} and $(l_j, \varphi_j) \in I$ iff $y_j$ is assigned \emph{True}.

The sets $\mathcal{I}$ and $\mathcal{S}$ are used both in Algorithm~\ref{alg:msr-scheme} and Algorithm~\ref{alg:mir-scheme}, and in both cases, they are gradually maintained during the whole computation of the algorithms.
To represent $\mathcal{I}$, we build a Boolean formula $\mathbb{I}$ such that a reduction $N$ \textbf{does not} belong to $\mathcal{I}$ iff $N$ \textbf{does} correspond to a model of $\mathbb{I}$.
Initially, $\mathcal{I} = \emptyset$, thus $\mathbb{I} = \mathit{True}$. To add an insufficient reduction $\TA_{<D,I>}$ and all reductions $N$, $N \sqsubseteq \TA_{<D,I>}$, to $\mathcal{I}$, we add to $\mathbb{I}$ the clause
$(\bigvee_{(e_i, \varphi_i) \in \Psi(\Delta) \setminus D} x_i) \vee (\bigvee_{(l_j, \varphi_j) \in \Psi(Inv) \setminus I} y_j)$.
To test if a reduction $N$ is in the set $\mathcal{I}$, we check if the valuation of $X \cup Y$ that corresponds to $N$ is not a model of $\mathcal{I}$.

Similarly, to represent $\mathcal{S}$, we build a Boolean formula $\mathbb{S}$ such that a reduction $N$ \textbf{does not} belong to $\mathcal{S}$ iff $N$ \textbf{does} correspond to a model of $\mathbb{S}$.
Initially, $\mathcal{S} = \emptyset$, thus $\mathbb{S} = \mathit{True}$. To add a sufficient reduction $\TA_{<D,I>}$ and all reductions $N$, $N \sqsupseteq \TA_{<D,I>}$, to $\mathcal{S}$, we add to $\mathbb{S}$ the clause
$(\bigvee_{(e_i, \varphi_i) \in D} \neg x_i) \vee (\bigvee_{(l_j, \varphi_j) \in I} \neg y_j)$.

The set $\mathcal{C}_s$ is used only in Algorithm~\ref{alg:msr-scheme}; namely in its subroutine {\findSSeed}. We build the set
$\mathcal{C}_s$ repeatedly during each call of {$\findSSeed(M, \mathcal{M}, \mathcal{I}, \mathcal{S})$} based on Equation~(\ref{eq:seed_candidates}) and we encode it via a Boolean formula $\mathbb{C}_s$ such that every model of $\mathbb{C}_s$ \textbf{does} correspond to a reduction $N \in \mathcal{C}_s$:
\begin{equation}
\mathbb{C}_s = \mathbb{I} \wedge \mathtt{trues( |M| -1 }\mathtt{)} 
\end{equation}
where $\mathtt{trues( |M| -1 }\mathtt{)}$ is a cardinality encoding forcing that exactly $|M| - 1$ variables from $X \cup Y$ are set to \emph{True}. To check if $\mathcal{C}_s = \emptyset$ or to pick a reduction $N \in \mathcal{C}_s$, we ask a SAT solver for a model of $\mathbb{C}_s$. 
To remove an insufficient reduction from $\mathcal{C}_s$, we update the formula $\mathbb{I}$ (and thus also $\mathbb{C}_s$) as described above.

Finally,
the set
$\mathcal{C}_i$ is used in the subroutine {\findISeed} of Algorithm~\ref{alg:mir-scheme}. We build the set repeatedly
during each call of {$\findISeed(M, \mathcal{M}, \mathcal{I}, \mathcal{S})$} and to represent it, we maintain a
Boolean formula $\mathbb{C}_i$ such that every model of $\mathbb{C}_i$ \textbf{does} correspond to a reduction $N \in \mathcal{C}_i$: 
\begin{equation}
\mathbb{C}_i = \mathbb{S} \wedge \mathtt{trues( |M| + 1 }\mathtt{)}
\end{equation}
where $\mathtt{trues( |M| + 1 }\mathtt{)}$ is a cardinality encoding forcing that exactly $|M| + 1$ variables from $X \cup Y$ are set to \emph{True}. To check if $\mathcal{C}_i = \emptyset$ or to pick a reduction $N \in \mathcal{C}_i$, we ask a SAT solver for a model of $\mathbb{C}_i$, and to remove a sufficient reduction from $\mathcal{C}_i$, we update the formula $\mathbb{S}$.


\section{Relaxing Minimal Sufficient Reductions}\label{sec:msr-relax}

In Section~\ref{sec:msr}, we considered a timed automaton $\TA = (L, l_0, C, \Delta, Inv )$ and a set of its locations $L_T \subseteq L$, and we presented an efficient algorithm to find a sufficient reduction (see Definition~\ref{defn:msr}), i.e., a set of simple clock constraints  $D \subseteq \Psi(\Delta)$~\eqref{eq:guard_cons_set}  (over transitions) and $I \subseteq \Psi(Inv)$~\eqref{eq:inv_cons_set}  (over locations) such that $L_T$ is reachable when constraints $D$ and $I$ are removed from $\TA$. In other words, $L_T$ is reachable on $\TA_{<D,I>}$.  Here, instead of completely removing $D \cup I$, our goal is to find a relaxation valuation $\mathbf{r} : D \cup I \to \mathbb{N}\cup\{\infty\}$ such that $L_T$ is reachable on $\TA_{<D,I,\mathbf{r}>}$. In addition, we intend to minimize the total change in the timing constants, i.e., $\sum_{\phi \in D \cup I} \mathbf{r}(\phi)$.
We present two methods to find such a valuation. The first one solves an MILP using a witness path $\pi'_{L_T}$ of $\TA_{<D,I>}$ that ends in $L_T$. The second one parametrizes each constraint from $D \cup I$ and solves a parameter synthesis problem on the resulting parametric timed automata. While the second method assumes all witness paths of $\TA_{<D,I>}$ and hence it is guaranteed to find the relaxation $\mathbf{r}$ with minimal $\sum_{\phi \in D \cup I} \mathbf{r}(\phi)$ for the considered MSR, the first method is computationally more efficient.

\subsection{MILP Based Relaxation}

By the definition of a sufficient reduction, the set $L_T$ is reachable on $\TA_{<D,I>}$.
Consequently, when a verifier is used to check the reachability of $L_T$, it
 generates a finite witness run $\rho'_{L_T}= (l_0, \textbf{0})  {\to}_{d_0} (l_1, v_1) {\to}_{d_1} \ldots  {\to}_{d_{n-1}} (l_n, v_n)$
 of $\TA_{<D,I>}$ such that $l_n \in L_T$.  Let $\pi'_{L_T} = l_0, e'_1, l_1, \ldots, e'_{n-1}, l_n$ be the corresponding path on $\TA_{<D,I>}$, i.e., $\pi'_{L_T}$ is realizable on $\TA_{<D,I>}$ due to the delay sequence $d_0, d_1, \ldots, d_{n-1}$ and the resulting run is $\rho'_{L_T}$.  The corresponding path on the original TA $\TA$ is defined in~\eqref{eq:path_map}:
\begin{equation}\label{eq:LTpath} \pi'_{L_T} = M(\pi_{L_T}), \text{ and }  \pi_{L_T} =  l_0, e_1, l_1, \ldots, e_{n-1}, l_n,
\end{equation}
While $\pi'_{L_T}$ is realizable on $\TA_{<D,I>}$, $\pi_{L_T}$ is not realizable on $\TA$ since $L_T$ is not reachable on $\TA$.  We present an MILP based method to find a relaxation valuation $\mathbf{r} : D \cup I \to \mathbb{N}\cup\{\infty\}$ such that the path induced by $\pi_{L_T}$ is realizable on $\TA_{<D,I,\mathbf{r}>}$.

For a given automaton path  $\pi = l_0, e_1, l_1, \ldots, e_{n-1}, l_n$ with $e_i = (l_{i-1}, \lambda_i, \phi_i, l_i)$ for each $i=1,\ldots,n-1$, we introduce real valued delay variables $\delta_0, \ldots, \delta_{n-1}$ that represent the time spent in each location along the path except the last one ($l_n$).
For a particular path, the value of a clock on a given constraint (invariant or guard) can be mapped to a sum of delay variables as each clock measures the time passed since its last reset:
\begin{equation}\label{eq:transform}
\Gamma(x, \pi, i) = \delta_k + \delta_{k+1} + \ldots  + \delta_{i-1}  \text{ where }  k = \max(\{ m  \mid x \in \lambda_m, m < i\} \cup \{ 0 \})
\end{equation}
The value of clock $x$ equals to $\Gamma(x, \pi, i)$ on the i-th transition $e_i$ along $\pi$. In~\eqref{eq:transform}, $k$ is the index of the transition where $x$ is last reset before $e_i$ along $\pi$, and it is $0$ if it is not reset. $\Gamma(0, \pi, i)$ is defined as $0$ for notational convenience.

Next, we define an MILP~\eqref{eq:LP} for the path $\pi$. By using the transformation~\eqref{eq:transform}, we map each clock constraint along the given path $\pi$ to constraints over the sequence of delay variables $\delta_0,\ldots, \delta_{n-1}$ as shown in~\eqref{eq:MILP_guard},\eqref{eq:MILP_a_invariant},\eqref{eq:MILP_l_invariant}. In addition, we introduce integer valued constraint relaxation variables  $\{ p_{l,\varphi} \mid (l,\varphi) \in I \}$ and $\{ p_{e,\varphi} \mid (e,\varphi) \in D\}$ for each simple constraint from $D \cup I$. In particular, for each transition $e_i$, the simple constraints $\varphi = x - y \sim c \in \mathcal{S}(\phi_i)$ of the guard $\phi_i$ of $e_i$ are mapped to the new delay variables~\eqref{eq:MILP_guard}, where $p_{e_i, \varphi}$ is the integer valued relaxation variable if $(e_i, \varphi) \in D$, otherwise it is set to $0$. On the other hand, for each location $l_i$, the simple clock constraints $\varphi = x - y \sim c \in \mathcal{S}(Inv(l_i))$ of the invariant $Inv(l_i)$ of $l_i$ are mapped to arriving~\eqref{eq:MILP_a_invariant} and leaving~\eqref{eq:MILP_l_invariant} constraints over the delay variables. In~\eqref{eq:MILP_a_invariant} and~\eqref{eq:MILP_l_invariant}, $\textbf{I}$ is a binary function mapping $\mathit{true}$ to $1$ and $\mathit{false}$ to 0, and $p_{l_i, \varphi_i}$ is the integer valued variable if $(l_i, \varphi_i) \in I$, otherwise it is set to $0$ as in~\eqref{eq:MILP_guard}. Note that if the invariant is satisfied when arriving and leaving, then, due to the convexity of the constraints, it is satisfied at every time when $\TA$ is at the corresponding location along $\pi$.
\begin{align}
& \text{minimize }  \sum_{(l,\varphi) \in I} p_{l,\varphi} +   \sum_{(e,\varphi) \in D} p_{e,\varphi}   \quad\quad\quad  \text{ subject to }\label{eq:LP}\\
& \Gamma(x, \pi, i) - \Gamma(y, \pi, i) \sim c +  p_{e_i, \varphi}   \quad\quad\quad \quad\quad\quad\quad \quad\quad \quad\quad\quad\quad \quad\quad\quad  (guard)  \nonumber \\
&  \quad\quad\quad \quad\quad\quad \quad\quad\quad  \text{ for each }i = 1, \ldots, n-1,  \text{ and } \varphi = x - y \sim c \in\mathcal{S}(\phi_i)\label{eq:MILP_guard} \\
& \Gamma(x, \pi, i) \cdot \textbf{I}(x \not \in \lambda_i) - \Gamma(y, \pi, i) \cdot \textbf{I}(y \not \in \lambda_i) \sim c + p_{l_i, \varphi}\quad (arriving, invariant) \nonumber \\
& \quad\quad\quad\quad\quad\quad\quad\quad\quad\quad \text{ for each } i = 1, \ldots, n, \varphi=x - y \sim c \in \mathcal{S}(Inv(l_i))\label{eq:MILP_a_invariant}\\
& \Gamma(x, \pi, i+1) - \Gamma(y, \pi, i+1) \sim c + p_{l_i, \varphi} \quad\quad\quad\quad\quad \quad\quad  (leaving, invariant) \nonumber \\
& \quad\quad\quad\quad\quad\quad\quad\quad \text{ for each } i = 0, \ldots, n-1, \varphi=x - y \sim c \in \mathcal{S}(Inv(l_i))\label{eq:MILP_l_invariant} \\
& p_{l,\varphi} \in \mathbb{Z}_+ \quad\quad\quad\quad\quad\quad\quad\quad\quad\quad\quad\quad\quad\quad\quad\quad\quad\quad\quad \text{ for each } (l,\varphi) \in I \\
& p_{e,\varphi} \in \mathbb{Z}_+ \quad\quad\quad\quad\quad\quad\quad\quad\quad\quad\quad\quad\quad\quad\quad\quad\quad\quad\quad \text{ for each } (e,\varphi) \in D    \\
& \delta_i \geq 0  \quad  \quad\quad\quad\quad\quad\quad\quad\quad\quad\quad\quad\quad\quad\quad\quad\quad\quad  \text{ for each } i=0,\ldots,n-1
\end{align}
Let $\{ p^\star_{l,\varphi} \mid (l,\varphi) \in I \}$, $\{ p^\star_{e,\varphi} \mid (e,\varphi) \in D\}$, and $\delta^\star_0,\ldots, \delta^\star_{n-1}$ denote the solution of MILP~\eqref{eq:LP}.
Define a relaxation valuation $\mathbf{r}$ with respect to the solution as
\begin{equation}\label{eq:opt_relaxation}
\mathbf{r}(l,\varphi) = p^\star_{l,\varphi} \text{ for each } (l,\varphi) \in I,    \quad \mathbf{r}(e,\varphi) = p^\star_{e,\varphi} \text{ for each } (e,\varphi) \in D.
\end{equation}

\begin{thm}~\label{thm:LPsoln} Let $\TA = (L, l_0, C, \Delta, Inv )$ be a timed automaton, $\pi = l_0, e_1, l_1, \ldots, e_{n}, l_n$ be a finite path of $\TA$, and $D \subseteq \Psi(\Delta)$, $I \subseteq \Psi(I)$ be guard and invariant constraint sets. If the MILP constructed from $\TA$, $\pi$, $D$ and $I$ as defined in~\eqref{eq:LP} is feasible, then $l_n$ is reachable on $\TA_{<D,I,\mathbf{r}>}$ with $\mathbf{r}$ as defined in~\eqref{eq:opt_relaxation}.
\end{thm}

\begin{proof}
Denote the optimal solution of MILP~\eqref{eq:LP} by $\{ p^\star_{l,\varphi} \mid (l,\varphi) \in I \}$, $\{ p^\star_{e,\varphi} \mid (e,\varphi) \in D\}$, and $\delta^\star_0,\ldots, \delta^\star_{n-1}$. For simplicity of presentation set $p^\star_{l,\varphi}$ to $0$ for each $(l,\varphi) \in \Psi(Inv) \setminus I$ and set $p^\star_{e,\varphi}$ to $0$ for each $(e,\varphi) \in \Psi(\Delta) \setminus D$. Let $\TA_{<D,I,\mathbf{r}>} = (L, l_0, C, \Delta', Inv' )$ and $T(\TA_{<D,I,\mathbf{r}>}) = (S, s_0, \Sigma, \to)$. Define clock value sequence $v_0, v_1, \ldots, v_{n}$ with respect to the path $\pi$ with $e_i = (l_{i-1}, \lambda_i, \phi_i, l_i)$ and the delay sequence $\delta^\star_0,\ldots, \delta^\star_{n-1}$ iteratively as $v_i=\textbf{0}$ and $v_i = (v_{i-1} + \delta^\star_{i-1})[\lambda_{i} := 0]$  for each  $i=1,\ldots, n$. Along the path $\pi$, $v_i$ is consistent  with $\Gamma(\cdot, \pi, i)$~\eqref{eq:transform} such that \begin{equation}\label{eq:mapvd} (a)\ v_i(x) = \Gamma(x, \pi, i).I(x \not \in \lambda_i) \quad \quad and  \quad (b)\ v_i(x) + \delta^\star_i = \Gamma(x, \pi, i+1)
\end{equation}
For a simple constraint $\varphi = x - y \sim c + p^\star_{l_i,\varphi} \in Inv'(l_i)$ (i.e. $x - y \sim c \in Inv(l_i)$ via Definition~\ref{defn:relaxation} and~\eqref{eq:opt_relaxation}), it holds that $v_i(x) - v_i(y) \sim c + p^\star_{l_i,\varphi} $ via~\eqref{eq:MILP_a_invariant} and~\eqref{eq:mapvd}-$a$.
Then by~\eqref{eq:opt_relaxation} $v_i \models Inv'(l_i)$ and $(l_i, v_i) \in S$. Similarly, $v_i + \delta^\star_i \models Inv'(l_i)$ via~\eqref{eq:MILP_l_invariant} and~\eqref{eq:mapvd}-$b$.
Hence, $(l_i, v_i + \delta^\star_i) \in S$ and $(l_i, v_i) \stackrel{\delta^\star_i}{\to} (l_i, v_i + \delta^\star_i)$ (delay transition).
Furthermore, by~\eqref{eq:MILP_guard},~\eqref{eq:opt_relaxation},~\eqref{eq:mapvd}-$b$ and Definition~\ref{defn:relaxation}, we have $v_i + \delta^\star_i \models R(\phi_i, D|_{e_i} , \mathbf{r}|_{e_i})$ and $(l_i, v_i + \delta^\star_i) \stackrel{act}{\to} (l_{i+1}, v_{i+1})$ (discrete transition). As $s_0 = (l_0, \textbf{0}) \in S$, and the derivation applies to each $i=1,\ldots, n$, we reach that $\rho = (l_0, v_0),\ldots, (l_n, v_n) \in [[\TA_{<D,I,\mathbf{r}>}]]$, and $l_n$ is reachable on $\TA_{<D,I,\mathbf{r}>}$.
\end{proof}

A linear programming (LP) based approach was used in~\cite{Bouyer2007OnTO} to generate the optimal delay sequence for a given path of a weighted timed automata. In our case, the optimization problem is in MILP form since we find an integer valued relaxation valuation ($\mathbf{r}$) in addition to the delay variables.

Recall that we construct relaxation sets $D$ and $I$ via Algorithm~\ref{alg:msr-scheme}, and define $\pi_{L_T}$~\eqref{eq:LTpath} that reach $L_T$ such that the corresponding path $\pi'_{L_T}$ is realizable on $\TA_{<D,I>}$. Then, we define MILP~\eqref{eq:LP} with respect to $\pi_{L_T}$, $D$ and $I$, and define $\mathbf{r}$~\eqref{eq:opt_relaxation} according to the optimal solution. Note that this MILP is always feasible since $\pi'_{L_T}$ is realizable on $\TA_{<D,I>}$. Finally, by Theorem~\ref{thm:LPsoln}, we conclude that $L_T$ is reachable on $\TA_{<D,I,\mathbf{r}>}$.

\subsection{Parameter Synthesis Based Relaxation}\label{sec:msr-relax-parametric}

As our second approach, we parametrize each simple constraint in the considered MSR\@. In particular, for each $(v, \varphi = x - y \sim c) \in D \cup I$ ($v$ is either a transition $e$ or a location $l$), we introduce a positive valued parameter $p_{v,\varphi}$ and replace the corresponding constraint with  $x - y \sim c + p_{v,\varphi}$. The resulting TA $\TA^{D\cup I}$ is parametric with parameter set $P = \{p_{(v,\varphi)} \mid (v, \varphi) \in D \cup I\}$. $\TA^{D\cup I}$ has $|D \cup I|$ parametric constraints and each parameter appears in a single constraint.  Subsequently, we use a parameter synthesis tool that generates the set of all parameter valuations $\mathbb{P} \subseteq \mathbb{R}_+^{\mid P \mid}$ for $\TA^{D\cup I}$ such that the target set $L_T$ becomes reachable, i.e., for each $\mathbf{p} \in \mathbb{P}$, $L_T$ is reachable on $\TA^{D\cup I}(\mathbf{p})$, where $\TA^{D\cup I}(\mathbf{p})$ is a non-parametric TA obtained from $\TA^{D\cup I}$ and $\mathbf{p}$ by replacing each parameter $p_{v,\varphi}$ with the corresponding valuation $\mathbf{p}(p_{v,\varphi})$. Then, we choose the integer valued parameter valuation $\mathbf{p}^\star : P  \to \mathbb{N}$ that minimizes the total change, i.e,  $\mathbf{p}^\star = \arg\min_{\mathbf{p} \in \mathbb{P} \cap \mathbb{N}} \sum_{(v,\varphi) \in D \cup I} \mathbf{p}(p_{v,\varphi})$. The parameter synthesis method ensures that $L_T$ is reachable on $\TA_{<D,I,\mathbf{r}>}$, where $\mathbf{r}$ is defined from $\mathbf{p}^\star$ as in~\eqref{eq:opt_relaxation}.

\subsection{Comparison of the MSR Relaxation Methods} The MILP based relaxation method minimizes the total change in the timing constants ($\sum_{\phi \in D \cup I} \mathbf{r}(\phi)$) for a particular path $\pi_{L_T}$. Thus, the resulting relaxation valuation~\eqref{eq:opt_relaxation} is not necessarily minimal for the considered MSR $D \cup I$.
Whereas, the parameter synthesis based relaxation method is guaranteed to find the minimal valuation (as it considers all paths of $\TA_{<D,I>}$). However, it is computationally more expensive compared to the MILP approach due to the complexity of the parameter synthesis for timed automata.

Let us note that both our approaches work with a fixed minimum MSR $\TA_{<D,I>}$. However, observe that
there might exist another minimum MSR $\TA_{<D',I'>}$ with $|D' \cup I'| = |D \cup I|$ that would lead to a smaller overall change of the constraints (i.e., smaller $\sum_{c \in D' \cup I'} \mathbf{r}(c)$). While our approach can be applied to a number of minimum MSRs, processing all of them can be practically intractable.


\section{Relaxing Minimal Guarantees}\label{sec:mg-relax}

In Section~\ref{sec:mg}, we presented a method to find a minimal guarantee $D \cup I$ (MG), i.e., a minimal subset of the constraints that need to be left in the system to ensure that a target (unsafe) location is still not reachable. In particular,  $L_T$ is not reachable on $\TA' = \TA_{<\Psi(\Delta) \setminus D, \Psi(Inv) \setminus I>}$ (see Definition~\ref{defn:mg}). In this section, we attempt to relax the timing constraints in the resulting TA $\TA'$, i.e., $D \cup I$, as much as possible while ensuring that $L_T$ is still unreachable. Thus, we analyze how \textit{robust} the resulting TA is against constraint perturbations with respect to the safety specification.
We consider two settings for relaxing the constraints from the MG\@. First, as in the MSR case, we find the maximal total relaxation of the remaining clock constraints such that $L_T$ is still unreachable.  Second, we find a single relaxation value $\delta$ such that $L_T$ is still unreachable when each constraint is relaxed by $\delta$, that is referred as the robustness degree in literature~\cite{DBLP:conf/rp/BouyerMS13}.

\subsection{Maximizing the Total Change}\label{sec:mg-relax-total-change}

As described in Section~\ref{sec:msr-relax-parametric}, we parametrize each simple constraint from the considered constraint set, i.e.\ in this case, it is the MG $D \cup I$ on $\TA'=\TA_{<\Psi(\Delta) \setminus D, \Psi(Inv) \setminus I>}$. Note that each constraint in the resulting TA that is denoted by $\TA^{D\cup I}$ is parametric  and the parameter set is $P = \{p_{(v,\varphi)} \mid (v, \varphi) \in D \cup I\}$. Then, we use a parameter synthesis tool that generates the set of all parameter valuations $\mathbb{P} \subseteq \mathbb{R}_+^{\mid P \mid}$ for $\TA^{D\cup I}$ such that the set $L_T$ is still unreachable. Finally, we chose the integer valued parameter valuation $\mathbf{p}^\star : P  \to \mathbb{N}$ that maximize the total change, i.e,  $\mathbf{p}^\star = \arg\max_{\mathbf{p} \in \mathbb{P} \cap \mathbb{N}} \sum_{(v,\varphi) \in D \cup I} \mathbf{p}(p_{v,\varphi})$. Note that, the maximal total change is finite since $\mathbf{p}(p_{v,\varphi})$ is finite for each valuation $\mathbf{p} \in \mathbb{P}$ and constraint $(v,\varphi) \in D \cup I$  due to the minimality of the MG\@.  The integer valued parameter valuation identifies the maximal total change in the constraint thresholds that can be applied to the TA $\TA_{<\Psi(\Delta) \setminus D, \Psi(Inv) \setminus I>}$ without violating the safety specification.  In particular, for any relaxation valuation $\mathbf{r}$ over $D \cup I$ with $\sum_{(v,\varphi) \in D \cup I} \mathbf{r}(v,\varphi) > \sum_{(v,\varphi) \in D \cup I} \mathbf{p}^\star(p_{v,\varphi})$,  the automaton $\TA'_{<D, I, \mathbf{r}>}$ violates the safety specification.

\subsection{Finding the Robustness Degree}\label{sec:mg-relax-delta}
A timed automaton $\TA = (L, l_0, C, \Delta, Inv )$ is said to $\delta$-robustly satisfy a linear-time property, such as a safety property, if the TA $\TA_{<\Psi(\Delta), \Psi(Inv), \mathbf{r}_\delta>}$ obtained by relaxing each simple constraint of $\TA$ by $\delta$ satisfies the property~\cite{DBLP:conf/rp/BouyerMS13,DBLP:journals/fmsd/WulfDMR08}, where
 \[\mathbf{r}_\delta(v,\varphi) = \delta \text{  for each } (v,\varphi) \in \Psi(\Delta) \cup \Psi(Inv).\]
A robustness value $\delta$ can be found via parametric analysis~\cite{DBLP:conf/rp/BouyerMS13,DBLP:conf/rp/AndreS11}.
Here, our goal is to find the maximal robustness value $\delta^\star$ for the timed automaton $\TA'$ such that $L_T$ is not reachable on $\TA'_{<D, I, \mathbf{r}_{\delta^\star}>}$ (recall that $\TA' = \TA_{<\Psi(\Delta) \setminus D, \Psi(Inv) \setminus I>}$).
Let $\mathbb{P}$ be the parameter valuation set defined as in Section~\ref{sec:mg-relax-total-change}.  Then, $L_T$ is not reachable on $\TA'_{<D, I, \mathbf{r}_{\delta^\star}>}$ for each $\delta \in \mathbf{D}$ ($\TA'$ $\delta-$robustly satisfies the safety specification), where
\[\mathbf{D} = \{ \delta \mid \textbf{p} \in \mathbb{P} \text{ and } \textbf{p}(v,\varphi) \geq \delta \text{ for each } (v,\varphi) \in D \cup I\} \]
 Alternatively, one can use the same parameter $p$ for each simple constraint of $\TA'$ to obtain a parametric TA $\TA'^{D \cup I}$ from $\TA'$ by replacing each simple constraint $(v, x - y \sim c) \in D \cup I$ with $(v, x - y \sim c + p)$. The resulting TA $\TA'^{D \cup I}$ has a single parameter $p$ and $\mid D \cup I \mid$ parametric constraints. Then, a parameter synthesis tool generates the set of all parameter valuations for $\TA'^{D \cup I}$ such that the set $L_T$ is still unreachable. Note that the set obtained in the second case is equal to $\mathbf{D}$.

In literature, the robustness analysis is studied considering the imperfect implementations of the $\TA$, e.g, timing or measuring errors, thus real valued robustness is used. In this work, we analyze the properties of the timed automata model itself, i.e., constraints and the constraint thresholds. Hence, we focus on integer valued relaxations of the TA\@. For this reason, we define the optimal relaxation value as $\delta^\star = \max \mathbf{D} \cap \mathbb{N}$.


\section{Related Work}\label{sec:related}
\subsection{Timed Automata}
In the literature, the uncertainties about timing constants are handled by representing such constants as parameters in a parametric timed automaton (PTA), i.e., a TA where clock constants can be represented with parameters. Subsequently, a parameter synthesis method, such as~\cite{imitatorTool,DBLP:conf/tacas/LimeRST09,DBLP:journals/jlp/BezdekBCB18}, is used to find suitable values of the parameters for which the resultant TA satisfies the specification. However, most of the parameter synthesis problems are undecidable~\cite{10.1007/s10009-017-0467-0}. While symbolic algorithms without termination guarantees exist for some subclasses~\cite{DBLP:conf/rp/AndreS11,10.1007/978-3-319-41591-8_12,6895298,Andr2019ParametricVA}, these algorithms are computationally very expensive compared to model checking (see~\cite{10.1007/978-3-030-12988-0_5}). Furthermore, it is not straightforward to integrate the minimization of the number of modified constraints in the parameter synthesis method for the reachability properties in an efficient way. For example, Imitator tool~\cite{imitatorTool} generates all parameter valuations such that the reachability or the safety property holds when the synthesis algorithm terminates. One approach would be parametrizing each simple constraint of the TA, then finding the valuation minimizing the number of non-zero parameters returned by the tool for the reachability problem. However, due to the dependence of the computation time on the number of parameters, this approach would be impractical. Similarly, for the safety problem, each constraint can be parametrized and further analysis can be performed on the result returned by the synthesis tool in order to find the minimal set of constraints that need to be left in the TA to ensure safety.  While assigning $0$ to a parameter that bounds a clock from below (i.e. $p \leq x$) or infinity to a parameter that bounds a clock from above (i.e. $x < p$) are equivalent to removing these constraints, it is not straightforward to deduce the constraint removal decision for constraints that involve multiple clocks (i.e $x - y \leq p$). Moreover, as mentioned for reachability, it would be impractical to solve the parameter synthesis problem when each constraint is parametrized.

Repair of a TA has been studied in recent works~\cite{10.1007/978-3-030-25540-4_5,ErgurtunaEarlyaccess,repairPaper}, where, similar to the reachability problem considered in this paper, the goal is to modify a given timed automaton such that the repaired TA satisfies the specification. In~\cite{repairPaper}, it is assumed that some of the clock constraints are incorrect and the goal is to make the TA compliant with an oracle that decides if a trace of the TA belongs to a system or not. To repair the TA, the authors of~\cite{repairPaper} parametrize the initial TA and generate parameters by analyzing traces of the TA\@. They minimize the total change of the timing constraints, while we primarily minimize the number of changed constraints and then the total change. Furthermore, their approach cannot handle reachability properties.
In~\cite{10.1007/978-3-030-25540-4_5,ErgurtunaEarlyaccess}, the goal is to repair the TA to avoid undesired behaviors, e.g., traces violating universal properties such as safety. In particular, in~\cite{10.1007/978-3-030-25540-4_5}, a single violating trace is analyzed by running an SMT solver on a linear arithmetic encoding of the trace. The generated repair suggestions include introducing clock resets and changing the clock constraints (both constraint bounds and constraint operators). As these operations can significantly change the set of traces of the automaton, they check the equivalence of the original and the repaired models after applying the suggested repair. In~\cite{ErgurtunaEarlyaccess}, new clocks and constraints over these new clocks are introduced to restrict the behavior of the automaton to eliminate the violating traces. Neither of these approaches can handle reachability properties. For safety properties, we consider a timed automaton satisfying the property, identify the constraints of the automaton that are effective in the satisfaction of the property and further analyze these constraints. On the other hand, both~\cite{10.1007/978-3-030-25540-4_5} and~\cite{ErgurtunaEarlyaccess} aim at repairing a TA that violates the given property.

The robustness of timed automata is studied considering non-ideal implementations of the model, i.e., imprecise clocks, measuring errors, etc.~\cite{DBLP:conf/rp/BouyerMS13,DBLP:journals/fmsd/WulfDMR08}. A timed automaton is said to be robust against clock perturbations and drifts for safety specifications when a TA obtained by allowing the clocks to drift within the given limits and relaxing each constraint by a certain amount satisfies the specification.
A complementary approach to robustness analysis is called shrinkability~\cite{DBLP:conf/fsttcs/SankurBM11,DBLP:conf/cav/Sankur13}: tighten (shrink) all of the constraints by a positive amount while guaranteeing that the resulting automaton is non-blocking and/or time abstract simulates the original one (thus preserves the safety and reachability properties).
Consequently, the shrunk automaton is robust against constraint perturbations. Region automata construction and difference bound matrices are used for the computation of the robustness degree in~\cite{DBLP:conf/rp/BouyerMS13,DBLP:journals/fmsd/WulfDMR08,DBLP:conf/fsttcs/SankurBM11,DBLP:conf/cav/Sankur13}. A parameter synthesis method is also utilized to find the robustness in~\cite{imitatorTool}. In this work, a similar constraint relaxation approach is used for reachability and safety specifications. To satisfy reachability specifications, we relax the constraints from minimal sufficient reductions. For safety specifications, we first identify a set of constraints that are active in satisfying the safety specification (minimal guarantee, MG), and then perform robustness analysis only over these constraints. In order to relax the identified constraints, we present an MILP based approach and also employ parameter synthesis by parametrizing constraints from the identified sets.

\subsection{Minimal Sets over a Monotone Predicate}
Although the concepts of minimal sufficient reductions (MSRs) and minimal guarantees (MGs) are novel in the context of timed automata,
similar concepts appear in other areas of computer science.
For example, see
minimal unsatisfiable subsets~\cite{DBLP:conf/ppdp/BandaSW03}, minimal correction subsets~\cite{DBLP:conf/ijcai/Marques-SilvaHJPB13}, minimal inconsistent subsets~\cite{looney,issta}, or minimal inductive validity cores~\cite{DBLP:conf/fmcad/GhassabaniWG17}.
All these concepts can be generalized as \emph{minimal sets over monotone predicates (MSMPs)}~\cite{DBLP:conf/cav/Marques-SilvaJB13,DBLP:journals/ai/Marques-SilvaJM17}. The input is a reference set $R$ and a monotone predicate $\mathbf{P}: \mathcal{P}(R) \rightarrow \{ 1, 0 \}$, and the goal is to find minimal subsets of $R$ that satisfy the predicate.
In the case of MSRs, the reference set is the set of all simple constraints $\Psi(\Delta) \cup \Psi(Inv)$ and, for every $D \cup I \subseteq \Psi(\Delta) \cup \Psi(Inv)$, the predicate is defined as $\mathbf{P}(D \cup I) = 1$ iff $\TA_{<D,I>}$ is sufficient.
Similarly, in the case of MGs,
the reference set is the set of all simple constraints $\Psi(\Delta) \cup \Psi(Inv)$ and, for every $D \cup I \subseteq \Psi(\Delta) \cup \Psi(Inv)$, the predicate is defined as $\mathbf{P}(D \cup I) = 1$ iff $\TA_{<\Psi(\Delta) \setminus D, \Psi(Inv) \setminus I>}$ is insufficient.

Many algorithms for finding MSMPes were proposed (e.g.,~\cite{DBLP:conf/cp/IgnatievPLM15,DBLP:journals/constraints/LiffitonMLAMS09,marco,mcsmus,tome,remus,unimus,DBLP:conf/ijcai/Marques-SilvaHJPB13,rime,DBLP:journals/constraints/IvriiMMV16,DBLP:conf/fmcad/GhassabaniWG17,DBLP:conf/sefm/BendikGWC18}), including also several algorithms (e.g.~\cite{DBLP:conf/cp/IgnatievPLM15,DBLP:journals/constraints/LiffitonMLAMS09,DBLP:journals/constraints/IgnatievJM16}) for extracting minimum MSMPs.
Most of the existing algorithms are  \emph{domain-specific}, i.e.\ tailored to a particular instance of MSMP and extensively exploiting specific properties of the instances (such as we exploit reduction cores in case of MSRs). Hence, the domain-specific solutions cannot be directly used for finding MSRs and/or MGs.
Several \emph{domain-agnostic} MSMP identification algorithms (e.g.~\cite{daa,pdds,marco}) were also proposed, i.e., algorithms that can be used for any type of MSMPs. Due to their universality, domain-agnostic approaches are usually not as efficient as the domain-specific solutions.
However, it is often the case that a domain-agnostic algorithm serves as a basis while building a domain-specific solution~\cite{DBLP:conf/lpar/BendikC18,dissertation}.
Some techniques we presented in this paper, including mainly the symbolic representation (Section~\ref{sec:representation}) and  the shrinking and growing procedures, are specializations of existing domain-agnostic solutions (see~\cite{marco, dissertation}).


\section{Experimental Evaluation}\label{sec:experiments}

We implemented the proposed reduction, guarantee and relaxation methods in a tool called Tamus.
We use UPPAAL~\cite{UPPAAL4:0} for sufficiency checks and witness computation, Imitator~\cite{imitatorTool} for parameter synthesis for PTA and CBC solver from Or-tools library~\cite{ortools} for the MILP part. All experiments were run on a laptop with Intel i5 quad core processor at 2.5 GHz and 8 GB ram using a time limit of 20 minutes per benchmark. The tool and used benchmarks are available at
\url{https://github.com/jar-ben/tamus}.

As discussed in Section~\ref{sec:related}, an alternative approach to solve the MSR problem
(Problem~\ref{prob:msr-relax}) is to
parameterize each simple clock constraint of the TA\@. Then, we can run a parameter synthesis tool on the parameterized TA
to identify the set of all possible valuations of the parameters for which the TA satisfies the reachability property.
Subsequently, we can choose the valuations that assign non-zero values (i.e., relax) to the minimum number of parameters, and out of these, we can choose the one with a minimum cumulative change of timing constants.
In our experimental evaluation, we evaluate the
state-of-the-art parameter synthesis tool Imitator~\cite{imitatorTool} to run such analysis. Although Imitator is not tailored for our problem, it allows us to measure the relative scalability of our approach compared to a well-established synthesis technique. In addition, we employ Imitator to solve the parameter synthesis problems for finding the optimal relaxation for a given MSR (Section~\ref{sec:msr-relax-parametric}), to find the maximal total change for a given MG (Section~\ref{sec:mg-relax-total-change}) and to find the robustness degree for the MG (Section~\ref{sec:mg-relax-delta}).

We used two collections of benchmarks to evaluate the proposed methods: one is obtained from the literature, and the other are crafted timed automata modeling a machine scheduling problem. In the following, we introduce these benchmarks and present the results of the experiments for reductions and guarantees.

\subsection{Experimental Results on Machine Scheduling Automata}\label{sec:machineScheduling}

\begin{table}[t]
\centering

\begin{tabular}{  l  c  c  c  c  c  c  c  c  c  c  c  c  c }
\toprule
\multicolumn{2}{c}{Model}&\multicolumn{4}{c}{MSR Results} & \multicolumn{8}{c}{MG Results} \\
\cmidrule(lr){1-2} \cmidrule(lr){3-6} \cmidrule(lr){7-14}
Name & $|\Psi|$ & $ d_{R} $  & $v_{R}$ &  $t_{R}$ & $c_{R}$ & $ d_{G} $  & $v_{G}$ &  $t_{G}$ & $c_{G}$ & $t_G^{IT}$ & $\delta^\star_\mathbb{R}$ & $\delta^\star$ & $t_G^{ITS}$ \\
\midrule
$\TA_{(3,1,12)}$ & 11 & 2 & 33 & 0.18 & 6  & 5 & 88 & 0.55 & 2 & 0.006 & 0.59 & 0 & 0.003\\
$\TA_{(3,2,12)}$ & 17 & 1 & 13 & 0.13 & 13 & 9 & 175 & 1.47 & 14 & 0.016 & 0.59 & 0 & 0.005\\

$\TA_{(3,1,18)}$ & 16 & 3 & 61 & 0.40 & 9  & 5 & 279 & 1.95 & 2 & 0.006 & 0.59 & 0 & 0.005 \\
$\TA_{(3,2,18)}$ & 24 & 1 & 498 & 4.68 & 6 & 10 & 519 & 5.33 & 7 & 0.031 & 0.59 & 0 & 0.007 \\

$\TA_{(3,1,24)}$ & 21 & 4 & 96 & 0.73 & 12 & 5 & 985 & 8.15 & 2 & 0.005 & 0.59 & 0 & 0.002 \\
$\TA_{(3,2,24)}$ & 32 & 1 & 51 & 0.65 & 16 & 12 & 1291 & 16.76 & 0 & 0.051 & 0.59 & 0 & 0.008 \\

$\TA_{(3,1,30)}$ & 26 & 5 & 140 & 1.24 & 15 & 5 & 1829 & 17.49 & 2 & 0.007 & 0.59 & 0 & 0.003 \\
$\TA_{(3,2,30)}$ & 40 & 1 & 63 & 0.96 & 9 & 13 & 2901 & 45.54 & 0 & 0.176 & 0.59 & 0 & 0.009\\
\midrule
$\TA_{(5,1,12)}$ & 16 & 3 & 90 & 0.55 & 10 & 4 & 214 & 1.39 & 1 & 0.004 & 0.49 & 0 & 0.002\\
$\TA_{(5,2,12)}$ & 24 & 1 & 192 & 1.54 & 13 & 6 & 166 & 1.46 & 14 & 0.007 & 1.33 & 1 & 0.003 \\

$\TA_{(5,1,18)}$ & 23 & 4 & 149 & 1.04 & 16 & 4 & 399 & 3.02 & 1 & 0.004 & 0.49 & 0 & 0.004 \\
$\TA_{(5,2,18)}$ & 35 & 1 & 25 & 0.35 & 6 & 7 & 796 & 9.13 & 3 & 0.010 & 0.24 & 0 & 0.008 \\

$\TA_{(5,1,24)}$ & 31 & 6 & 327 & 2.67 & 24 & 4 & 1050 & 9.34 & 1 & 0.004 & 0.49 & 0 & 0.004 \\
$\TA_{(5,2,24)}$ & 47 & 2 & 373 & 4.86 & 31 & 8 & 2708 & 40.46 & 13 & 0.021 & 0.49 & 0 & 0.015 \\

$\TA_{(5,1,30)}$ & 39 & 7 & 571 & 5.54 & 29 & 4 & 1864 & 19.57 & 1 & 0.007 & 0.49 & 0 & 0.006 \\
$\TA_{(5,2,30)}$ & 59 & 2 & 624 & 9.45 & 17 & 9 & 1556 & 26.45 & 6 & 0.028 & 0.49 & 0 & 0.010 \\
\midrule
$\TA_{(7,1,12)}$ & 19 & 3 & 119 & 0.74 & 11 & 3 & 153 & 0.97 & 0 & 0.004 & 0.33 & 0 & 0.003 \\
$\TA_{(7,2,12)}$ & 28 & 1 & 70 & 0.62 & 13 & 5 & 247 & 2.16 & 10 & 0.005 & 0.33 & 0 & 0.003 \\

$\TA_{(7,1,18)}$ & 28 & 5 & 314 & 2.33 & 25 & 3 & 402 & 3.19 & 0 & 0.002 & 0.33 & 0 & 0.002 \\
$\TA_{(7,2,18)}$ & 42 & 1 & 175 & 2.00 & 6 & 6 & 528 & 6.45 & 3 & 0.010 & 0.33 & 0 & 0.009 \\

$\TA_{(7,1,24)}$ & 38 & 7 & 615 & 5.29 & 39 & 3 & 1442 & 14.28 & 0 & 0.002 & 0.33 & 0 & 0.002 \\
$\TA_{(7,2,24)}$ & 57 & 2 & 944 & 12.67 & 21 & 6 & 863 & 12.86 & 11 & 0.012 & 0.33 & 0 & 0.011 \\

$\TA_{(7,1,30)}$ & 48 & 10 & 1559 & 16.75 & 47 & 3 & 2302 & 26.82 & 0 & 0.008 & 0.33 & 0 & 0.003 \\
$\TA_{(7,2,30)}$ & 72 & 2 & 675 & 11.25 & 14 & 7 & 1295 & 23.43 & 4 & 0.021 & 0.33 & 0 & 0.015 \\
\bottomrule

\end{tabular}\caption{Results for the scheduler TA, where $|\Psi| = | \Psi(\Delta) \cup \Psi(I) | $ is the total number of  constraints, $d_R/d_G=| D \cup I|$ is the minimum MSR/MG size, $v_R/v_G$ is the number of reachability checks during minimum MSR/MG computation, $t_R/t_G$ is the computation time in seconds for minimum MSR/MG computation (including the reachability checks), $c_R$ is the optimal cost of~\eqref{eq:LP}, $c_G$ is the maximal total change $\sum_{(v,\varphi) \in D \cup I} \mathbf{p}^\star(p_{v,\varphi})$ for the MG, $\delta^\star_\mathbb{R}$ is the real valued maximal robustness value, and $\delta^\star$ is the integer valued optimal relaxation value. $t_G^{IT}$ and $t_G^{ITS}$ are the Imitator computation times for maximizing the total change and finding the robustness degree, respectively. }\label{table:generator}

\end{table}

A scheduler automaton is composed of a set of paths starting in location $l_0$ and ending in location $l_1$.  Each path $\pi = l_0 e_k l_k e_{k+1} \ldots l_{k+M-1} e_{k+M} l_1$ represents a particular scheduling scenario where an intermediate location, e.g. $l_i$ for $i=k, \ldots, k+M-1$, belongs to a unique path (only one incoming and one  outgoing transition). Thus, a TA that has $p$ paths with $M$ intermediate locations in each path has $M \cdot p + 2$ locations and $(M+1) \cdot p $ transitions.  Each intermediate location represents a machine operation, and periodic simple clock constraints are introduced to mimic the limitations on the corresponding durations. For example, assume that the total time to use machines represented by locations $l_{k+i}$ and $l_{k+i+1}$ is upper (or lower) bounded by $c$  for $i=0,2,\ldots, M-2$.
To capture such a constraint with a period of $t=2$, a new clock $x$ is introduced and it is reset and checked on every $t^{th}$ transition along the path, i.e., for every $m \in \{ i \cdot t + k \mid i \cdot t  \leq M-1 \}$, let $e_{m} = (l_{m},\lambda_{m},  \phi_m, l_{m+1})$, add $x$ to $\lambda_{m}$, set $\phi_m := \phi_m \wedge x \leq c$  ($ x \geq c$ for lower bound).  A periodic constraint is denoted by $(t , c, \sim)$, where $t$ is its period, $c$ is the timing constant, and ${\sim} \in \{<,\leq, >, \geq \}$. A set of such constraints are defined for each path to capture possible restrictions. In addition, a bound $T$ on the total execution time is captured with the constraint $x \leq T $ on transition $e_{k+M}$ over a clock $x$ that is not reset on any transition. A realizable path to $l_1$ represents a feasible scheduling scenario.
We have generated $24$ test cases. A test case $\TA_{(c, p, M)}$ represents a timed automaton with $c \in \{3,5,7\}$ clocks, and $p \in \{1,2\}$ paths with $M \in \{12,18,24,30\}$ intermediate locations in each path. $R_{c,i}$ is the set of periodic restrictions defined for the $i^{th}$ path of an automaton with $c$ clocks:
\begin{align*}
& R_{3,1} = \{ (2, 11, \geq), (3, 15, \leq) \} & R_{3,2} = \{ (4, 17, \geq), (5, 20, \leq) \} \\
& R_{5,1}= R_{3,1} \cup \{ (4, 21, \geq), (5, 25, \leq) \} & R_{5,2} = R_{3,2} \cup \{(8, 33, \geq), (9, 36, \leq) \} \\
& R_{7,1}= R_{5,1} \cup \{ (6, 31, \geq), (7, 35, \leq) \} & R_{7,2} = R_{5,2} \cup \{(12, 49, \geq), (12, 52, \leq) \}
\end{align*}
Note that $\TA_{(c, 2, M)}$ emerges from $\TA_{(c, 1, M)}$ by adding a path with restrictions~$R_{c,2}$.

\paragraph{MSR analysis.} A path to $l_1$ describes a scheduling scenario for a scheduler automaton ($\TA_{(c, p, M)}$). However, location $l_1$ is unreachable for  each  of the introduced automata. Thus, our goal is to find a realizable path to $l_1$ by performing a minimum amount of change. In order to achieve this, we define the target set as $L_T = \{l_1\}$ and run the developed MSR methods. The results obtained on the scheduler automata are summarized in Table~\ref{table:generator}.
Tamus solved all models and the longest computation time was $16.75$ seconds.  As expected, the computation time $t_R$ is depends on the number
 $|\Psi|$ of simple clock constraints in the model.

When each simple constraint is parametrized, Imitator solved $\TA_{(3,1,12)}$, $\TA_{(3,2,12)}$, $\TA_{(3,1,18)}$, and $\TA_{(5,1,12)}$ within $0.09$, $0.5$, $62$, and $71$ seconds, respectively, and timed-out for the other models. In addition, we run Imitator with a flag ``witness'' that terminates the computation when a satisfying valuation is found. The use of this flag reduced the computation time
 for the aforementioned cases, and it allowed to solve two more models: $\TA_{(3,2,18)}$ and $\TA_{(5,2,12)}$. However, using this flag,  Imitator often did not provide a solution that minimizes the number of relaxed simple clock constraints.

\paragraph{MG analysis.} We also run the developed methods to find MGs, the corresponding maximal total changes and robustness values over the scheduler automata models with $L_T = \{l_1\}$. The results are reported in Table~\ref{table:generator}. Tamus was also able to generate MG results for all models and $\TA_{(5,2,24)}$ took the longest with $40.46$ seconds. In the MG case, when any of the $|\Psi| - d_G + 1$ simple constraints are removed from the original scheduler automata $\TA_{(c, p, M)}$, $L_T$ becomes reachable.  In addition, we run Imitator to find the maximal total change (Section~\ref{sec:mg-relax-total-change}) and the robustness degree (Section~\ref{sec:mg-relax-delta}) for the identified MG $D \cup I$. Specifically we first ran Tamus on TA $\TA_{(c, p, M)}$ and then removed every constraint that is not in $D \cup I$ (i.e., obtained $\TA_{(c, p, M),  <\Psi(\Delta) \setminus D, \Psi(Inv) \setminus I>}$) and parameterized every constraint that is in $D \cup I$. A different parameter is used for every constraint to find the maximal total change and the same parameter is used for every constraint to find the robustness degree. Since $d_G$ is much smaller than $|\Psi|$, Imitator generated the results for every model within $0.2$ seconds for both parameter synthesis approaches. Both integer valued and real valued results for the parameters are reported in Table~\ref{table:generator}. Location $l_1$ becomes reachable when each simple constraint in $\TA_{(c, p, M),  <\Psi(\Delta) \setminus D, \Psi(Inv) \setminus I>}$ is relaxed by $\delta^\star + 1$.

\begin{table}[t]
\centering
\begin{tabular}{  c  c  c  c  c  c  }
\toprule
Model & Source & Spec. & $|\Psi|$ & $|\Psi^{u}|$ &  $m$  \\
\midrule
accel1000 &~\cite{DBLP:journals/corr/abs-1812-08940}\cite{ARCH15:Benchmarks_for_Temporal_Logic} & reach. & 7690 & 13 & 3 \\

CAS &~\cite{10.1007/978-3-642-38916-0_2} & reach. & 18 & 18 & 9 \\

coffee &~\cite{Andr2019ParametricVA} & reach. & 10  & 10 & 3  \\

Jobshop4 &~\cite{10.1007/3-540-44585-4_46} & reach. & 64 & 48 & 5  \\

Pipeline3-3 &~\cite{Knapik2010BoundedMC} & reach. & 41 & 41 & 12\\ 

RCP &~\cite{CollombAnnichini2001ParameterizedRA} & reach. & 42 & 42 & 11 \\

SIMOP3 &~\cite{Andre:2009} & reach. & 80 & 80  & 40 \\

Fischer &~\cite{10.1007/3-540-45319-9_14} & safety & 24 & 16 & 0  \\

JLR13-3tasks &~\cite{10.1007/978-3-642-36742-7_28}\cite{10.1007/978-3-319-17524-9_5} & safety & 42 & 36 & 0  \\ 

WFAS &~\cite{DBLP:journals/corr/BenesBLS15}\cite{10.1007/978-3-319-06410-9_44}& safety & 32 & 24 & 0 \\
\bottomrule

\end{tabular}\caption{Properties of the benchmarks, where $|\Psi|$ is as defined in Table~\ref{table:generator}, $|\Psi^u|$ is the number of constraints considered in the analysis and $m$ is the number of mutated constraints.}\label{table:benchmarks}
\end{table}

\subsection{Experimental Results on Benchmarks from Literature}\hspace{3pt}
We collected 10 example models from the literature that include models with a safety specification that requires avoiding a set of locations $L_T$, and models with a reachability specification with a set of target locations $L_T$. In both cases, the original models satisfy the given specification. Eight of the examples are networks of TAs, and while a network of TAs can be represented as a single product TA and hence our methods can handle it,  Tamus currently supports only MSR and MG computations for networks of TA, but not MILP relaxation. The properties of these models are summarized in Table~\ref{table:benchmarks}.

\paragraph{MSR analysis.} For the safety specifications, we define $L_T$ as the target set and apply our methods for MSRs.  Here, we find the minimal number of timing constants that should be changed to reach $L_T$, i.e., to violate the original safety specification. On the other hand, for reachability specifications, inspired by mutation testing~\cite{10.1007/978-3-642-38916-0_2}, we change a number of constraints on the original model so that $L_T$ becomes unreachable. The number of mutated constraints are shown in Table~\ref{table:benchmarks}.

The MSR results are shown in Table~\ref{table:msr_results}. Tamus computed a minimum MSR for all the models and also provided the MILP relaxation for the non-network models. Note that the bottleneck of our approach is the MSR computation and especially the verifier calls; the MILP part always took only few milliseconds (including models from Table~\ref{table:generator}), thus we believe that it would be also the case for the networks of TAs.
The base variant of Imitator that computes the set of all satisfying parameter valuations solved only 4 of the 10 models. When run with the early termination flag, Imitator solved 3 more models, however, as discussed above, the provided solutions might not be optimal.

We have also evaluated Imitator for parameter synthesis based relaxation  (Section~\ref{sec:msr-relax-parametric}). In particular, we first run Tamus to compute a minimum MSR $\TA_{<D,I>}$, then parameterized the constraints $D \cup I$ in the original TA $\TA$, and run Imitator on the parameterized TA\@.
In this case, Imitator solved 9 out of 10 models. Moreover, we have the guarantee that we found the optimal solution: the minimum MSR ensures that we relax the minimum number of simple clock constraints, and Imitator finds all satisfying parameterizations of the constraints hence also the one with minimum cumulative change of timing constants.

\begin{table}[t]
\centering
\begin{tabular}{  c  c  c  c  c  c  c  c  c  c  }
\toprule
Model & $ d_R $ & $v_R$ &  $t_R$ & $c_R$ & $c_{I}$ & $t_R^I$ & $t_R^{IT}$ & $t_R^{Iw}$ & $t_R^{ITw} $ \\
\midrule
accel1000  & 2  & 22 & 1.50 & - & 5850 & 184.5 & 1.73 & 1.47 & 0.88 \\ 

CAS  & 2 & 46 & 0.36 & 16 & 16 & 0.80 & 0.10 & 0.06 & 0.02 \\

coffee  & 2  & 18 & 0.09 & 14 & 14 & 0.02 & 0.008 & 0.01 & 0.007 \\

Jobshop4  & 5  & 272 & 2.62 & - & 5 & to & 305.1 & to & 301.1 \\ 

Pipeline3-3  & 1  & 42 & 0.43 & - & 2 & to & 0.04 & to & 0.03 \\ 

RCP  & 1  & 99 & 1.69 & - & 253 & to & 0.04  & 34.38 & 0.03 \\ 

SIMOP3 & 6  & 833 & 14.24 & - & 3755 & to & 3.98 & to & 0.37 \\ 

Fischer & 1  & 14 & 0.09 & - & - & to & to & 0.17 & 0.03 \\ 

JLR13-3tasks & 1  & 40 & 0.51 & - & 20 & to & 3.24 & 0.31 & 0.10 \\ 

WFAS & 1  & 10 & 0.09 & - & 6 & 14.69 & 0.02 & 14.66 & 0.03 \\ 
\bottomrule

\end{tabular}\caption{Experimental MSR analysis results for the benchmarks, where $d_R$, $v_R$, $t_R$ and $c_R$ are as defined in Table~\ref{table:generator}. $c_{I}$ is the minimal total change $\sum_{(v,\varphi) \in D \cup I} \mathbf{p}^\star(p_{v,\varphi})$ for the MSR as explained in Section~\ref{sec:msr-relax-parametric}.
$t_R^I$, $t_R^{IT}$, $t_R^{Iw}$ and $t_R^{ITw}$ are the Imitator computation times, where $w$ indicates that the early termination flag (``witness'') is used, otherwise the largest set of parameters is searched, and $T$ indicates that only the constraints from the MSR identified by Tamus are parametrized, otherwise all constraints from $\Psi^u$ are parametrized. $to$ shows that the timeout limit is reached (20 min.).}\label{table:msr_results}
\end{table}

\paragraph{MG analysis.} In order to apply the developed methods for the guarantee sets, we use $L_T$ as the set of unsafe locations for both safety and reachability specifications. For the safety specifications, the minimal MG is the minimal set of constraints that are effective in avoiding the unsafe behaviors as intended. On the other hand, for the reachability specifications, as in the MSR analysis, we mutate a number of constraints on the original model so that $L_T$ becomes unreachable and then apply our MG methods.

The MG results are shown in Table~\ref{table:mg_results}. Tamus computed a minimum MG for all the models. The maximal total change and the robustness degrees are generated using Imitator. Imitator found robustness degrees for all of the models, whereas it created maximal total change results for 9 out of 10 models. The running time for Imitator for both parameter synthesis approaches is below 15 seconds for all the results it produced.

\begin{table}[t]
\centering
\begin{tabular}{  c  c  c  c  c  c  c  c  c }
\toprule
Model & $ d_G $ & $v_G$ &  $t_G$ & $c_G$  & $t_G^{IT}$ & $\delta^\star_\mathbb{R}$ & $\delta^\star$ & $t_G^{ITS}$\\
\midrule
accel1000 & 1 & 23 & 1.84 & 1557 & 1.3 & 1557.99 & 1557 & 1.2\\

CAS & 2 & 39 & 0.30 & 11 & 0.004 & 5.99 & 5 & 0.003\\

coffee & 2 & 25 & 0.12 & 7 & 0.005 & 3.99 & 3 & 0.002\\

Jobshop4 & 2 & 281 & 2.94 & 0 & 4.44 & 0.49 & 0 & 3.47\\

Pipeline3-3 & 2 & 82 & 0.67 & 1 & 0.04 & 0.99 & 0 & 0.03\\ 

RCP & 8 & 255 & 4.46 & 463 & 0.32 & 8.99 & 8 & 0.03\\

SIMOP3 & 2 & 214 & 3.81 & 417& 0.003 & 208.99 & 208 & 0.002\\

Fischer & 4 & 46 & 0.27 & 2 & 0.05 & 0.50 & 0 & 0.03\\

JLR13-3tasks & 10 & 101 & 1.26 & - & to & 7.99 & 7 & 14.12\\ 

WFAS & 5 & 66 & 0.47 & 0 & 0.06 & 0.00 & 0 & 0.018 \\
\bottomrule

\end{tabular}\caption{Experimental MG analysis results for the benchmarks, where $d_G$, $v_G$, $t_G$, $c_G$, $ \delta^\star_\mathbb{R}$, $ \delta^\star$, $t_G^{IT}$ and $t_G^{ITS}$ are as defined in Table~\ref{table:generator}.}\label{table:mg_results}
\end{table}

\section{Conclusion}\label{sec:conclusion}
We proposed the novel concept of a minimum MSR for a TA, i.e., a minimal set of simple clock constraints that need to be relaxed to satisfy a given specification. Moreover, we developed efficient techniques to find a minimum MSR, and presented MILP and parameter synthesis methods how to further tune the constraints in the MSR\@.
We also introduced the concept of a maximum MIR, i.e., a maximal set of simple clock constraints that can be removed from the TA without violating the specification. Dually, one can represent an MIR via its complementary MG, i.e., a minimal set of simple clock constraints that need to be left in the TA to ensure that the specification is not violated. Moreover, we proposed parameter synthesis based approaches that can
further relax the constraints in the MG while still keeping the specification satisfied.

Our empirical analysis showed that our tool, Tamus, can generate minimum MSRs and minimum MGs within seconds even for large systems.
For the task of MSR relaxation, we have shown that the MILP method is faster than the parameter synthesis approach (MSR + Imitator). However, the MILP approach minimizes the cumulative change of the constraints from a minimum MSR by considering a single witness path. If the goal is to find a minimal relaxation globally, i.e., w.r.t.\ all witness paths for the MSR, we recommend using the combined version of MSR and Imitator, i.e., first run Tamus to find a minimum MSR, parametrize each constraint from the MSR and run Imitator to
find all satisfying parameter valuations, including the global optimum.

\section*{Acknowledgment}
This research was supported in part by ERDF ``CyberSecurity, CyberCrime and Critical Information Infrastructures Center of Excellence'' (No. CZ.$02.1.01/0.0/$ $0.0/16\_019/0000822$) and in part by the European Union's Horizon 2020 research and innovation programme under the Marie Sklodowska-Curie grant agreement No. 798482.

\bibliography{ebru_t}
\bibliographystyle{alphaurl}

\end{document}